\documentclass[11pt]{article}
\usepackage[margin=1in]{geometry}
\usepackage{amsmath,amsthm,amssymb}
\usepackage{algorithm,algorithmic,soul}
\usepackage[fleqn,tbtags]{mathtools}
\usepackage{bbm}
\usepackage{color}
\usepackage{tikz}
\usepackage{subcaption}
\allowdisplaybreaks

\usepackage[round]{natbib} 

\newtheorem{lemma}{Lemma}[section]
\newtheorem{proposition}[lemma]{Proposition}

\newtheorem{corollary}[lemma]{Corollary}

\newtheorem{example1}[lemma]{Example}
\newtheorem{rem1}[lemma]{Remark}

\newtheorem{alg1}[lemma]{Algorithm}
\newtheorem{me1}[lemma]{Mechanism}
\newenvironment{remark}{\begin{rem1}\rm}{\end{rem1}}
\newenvironment{example}{\begin{example1}\rm}{\end{example1}}

\numberwithin{equation}{section}
\newcommand{\R}{\mathbb{R}}
\renewcommand{\r}{{\bf r}}
\newcommand{\sm}{\!\setminus\!}

\DeclareMathOperator{\cl}{cl}

\DeclareMathOperator{\co}{co}

\renewcommand{\succ}{\operatorname{succ}}
\DeclareMathOperator{\esssup}{ess\,sup}
\DeclareMathOperator{\essinf}{ess\,inf}

\renewcommand{\O}{\Omega}
\renewcommand{\o}{\omega}
\newcommand{\G}{\mathcal{G}}
\newcommand{\F}{\mathcal{F}}

\newcommand{\Y}{\mathcal{Y}}

\newcommand{\A}{\mathcal{A}}
\renewcommand{\L}{\mathcal{L}}

\renewcommand{\P}{\mathcal{P}}
\newcommand{\N}{\mathbb{N}}

\renewcommand{\Pr}{\mathbb{P}}
\newcommand{\Q}{\mathbb{Q}}

\newcommand{\1}{\mathbf{1}}

\newcommand{\of}[1]{\ensuremath{\left( #1 \right)}}
\newcommand{\cb}[1]{\ensuremath{ \left\{ #1 \right\} }}
\newcommand{\sqb}[1]{\ensuremath{ \left[ #1 \right] }}

\def\prehp(#1,#2){\ensuremath{  #1 \cdot #2 }}

\newcommand{\E}{\mathbb{E}}
\newcommand{\W}{\mathcal{W}}

\newcommand{\ycal}{\mathcal{Y}}

\newcommand{\lcal}{\mathcal{L}}
\newcommand{\T}{\top}

\newcommand{\ind}{\mathbbm{1}}

\begin{document}

\title{Set-Valued Risk Measures as Backward Stochastic Difference Inclusions and Equations}
\author{\c{C}a\u{g}{\i}n Ararat\thanks{Bilkent University, Department of Industrial Engineering, Ankara, Turkey, cararat@bilkent.edu.tr.}\and  Zachary Feinstein\thanks{Stevens Institute of Technology, School of Business, Hoboken, NJ, USA, zfeinste@stevens.edu.}
}
\date{September 13, 2020}
\maketitle
\abstract{
Scalar dynamic risk measures for univariate positions in continuous time are commonly represented as backward stochastic differential equations. In the multivariate setting, dynamic risk measures have been defined and studied as families of set-valued functionals in the recent literature. There are two possible extensions of scalar backward stochastic differential equations for the set-valued framework: (1) backward stochastic differential inclusions, which evaluate the risk dynamics on the selectors of acceptable capital allocations; or (2) set-valued backward stochastic differential equations, which evaluate the risk dynamics on the full set of acceptable capital allocations as a singular object.  In this work, the discrete time setting is investigated with difference inclusions and difference equations in order to provide insights for such differential representations for set-valued dynamic risk measures in continuous time.
}

\section{Introduction}\label{sec:intro}
\subsection{Literature Review}\label{sec:literature}
The seminal work of \cite{artzner} introduced coherent risk measures in an axiomatic framework to provide the capital requirements for financial portfolios.  The coherence axioms were relaxed to that of convexity in \cite{scalarshortfall,fs:sf}. Each of these works considers risk measured at a single time point for frictionless claims. These convex and coherent risk measures were placed in a dynamic system in which the capital requirements for portfolios and contingent claims are updated in time as new information becomes available. The relation of these risks over time is a key property of study; this so-called time consistency property has been studied in, e.g.,~\cite{R04,BN04,DS05,RS05,CDK06}.  

For the purposes of this work, the relevant literature is more specialized.  It has been noted for almost as long as dynamic risk measures have been studied that certain nonlinear ($g$-)expectations have all the properties of a time consistent convex risk measure, and vice versa.  In this way, convex risk measures can be represented as solutions to backward stochastic differential equations (see, e.g., \cite{peng19,BK04,BK09,G06,J08}). 
This representation allows for the detailed study of dynamic risk measures in continuous time as well as for their efficient computation. We take our motivation for this work from \cite{stadje}, which considers the problem of representing a \emph{discrete} time dynamic risk measure as the solution of a backward stochastic difference equation. That work further provides the $L^2$-convergence of the drivers and solutions of backward stochastic difference equations to the corresponding backward stochastic differential equation with desired driver as the number of time steps (in a finite time horizon) grows to infinity.

Scalar risk measures have been extended to consider multivariate portfolios as well.  Often, and with particular relevance to this work, they are studied as set-valued risk measures.  Such set-valued risk measures were first introduced in a coherent setting in \cite{JMT04}.  This was extended to consider convex risk measures in \cite{HHR}. Such risk measures have been applied, primarily, in two settings:
\begin{itemize}
\item \textbf{Portfolios with frictions}: due to transaction costs, valuing a portfolio by a num\'eraire leads to a non-unique value (e.g., mark-to-market or liquidation value).  Each choice of valuation by a num\'eraire results in an intrinsic misspecification of risk, for example: mark-to-market valuation provides an overestimate of the true value of the portfolio and thus an underestimate of the risk; liquidation valuation provides an underestimate of the true value of the portfolio and thus an overestimate of the risk. This is the prototypical case for set-valued risk measures as studied in~\cite{HHR}. Specifically, due to transaction and market-impact costs, there is a misspecification of value in using any num\'eraire to hedge a risky portfolio.  In order to compensate for these frictions, capital requirements in that specific num\'eraire must be larger than the current value of the hedging portfolio.  Set-valued risk measures approach this problem by considering the set of all risk-compensating portfolios, for which there is unlikely to be a uniquely minimal choice, so as to more accurately capture the riskiness of the original portfolio. In order to provide a num\'eraire-free representation, a portfolio is represented by a random vector in the physical units, i.e., the $i$th component of the random vector denotes the number of asset $i$ held in the portfolio.
\item \textbf{Systemic risk}: due to the interaction between banks, the capital requirements for any institution will ultimately depend on the health of all other institutions.  This is presented in a set-valued setting in \cite{syst,syst2}. Specifically, capital can\emph{not} be freely transferred between banks, but due to, e.g., default contagion (see~\cite{EN01}), the capital requirements for one institution intrinsically depend on the health of all other banks.  Due to this coupling of the health of all banks in the system, capital requirements for the different institutions cannot be determined independently.  In particular, this generates an efficient frontier of capital requirements for all banks in the financial system and, as such, such objects are naturally described by the set of all acceptable (system-wide) capital requirements. The collection of wealths for the financial system is presented as a random vector, i.e., the $i$th component of the random vector denotes the wealth of bank $i$ in the system.
\end{itemize}
There are analogs of many scalar risk measures in the set-valued framework.  For instance, superhedging was considered in \cite{LR14}, and shortfall and divergence risk measures were considered in \cite{sdrm}.

The typical analog of time consistency in the set-valued setting is ``multiportfolio time consistency'' as first defined in \cite{tc}. This was shown to be equivalent to a variation of the time consistency property in \cite{bentaharlepinette} for the random set approach in \cite{comparison}.  Additional properties for multiportfolio time consistency were studied in, e.g., \cite{mptc,supermtg,scalarmv,tcscalarmv}, and have been utilized directly for computing such risk measures in discrete time in \cite{algorithm}. Though not the focus of this work, multiportfolio time consistency has been extended to risk measures for processes in \cite{chenhu}.

Time consistency for risk measures, more generally, relates to the dynamic programming principle from optimization. Optimal control problems are typically formulated using backward stochastic differential equations in continuous time. In finance, the relation between time consistency, dynamic programming principle, and backward stochastic differential equations was highlighted in, e.g., \cite{riskarbitrage,timeinconsistency}. By considering multiportfolio time consistency, and the relation to the dynamic programming principle as espoused in~\cite{algorithm}, we wish to further investigate this relationship for set-valued processes.

\subsection{Motivation}\label{sec:motivation}
The backward stochastic differential equation (BSDE) representation of scalar risk measures in continuous time is fundamental to our understanding of how risks propagate over time. We briefly recall this representation. Under the augmented filtration of a standard $m$-dimensional Brownian motion $W=(W(t))_{t\in[0,T]}$, let us consider the solution $(Y^X,Z^X)$ of the BSDE
\begin{equation}\label{riskBSDE}
Y^X(t) = -X+\int_t^T g(s,Z^X(s))ds-\int_t^T Z^X(s)dW(s),\quad t\in[0,T],
\end{equation}
for each terminal value $X$ in the $L^2$-space of $\F_T$-measurable (univariate) random variables. As shown in \cite{G06}, the family of these solutions indexed by the terminal value gives rise to a time consistent dynamic risk measure via $X\mapsto \rho_t(X)\coloneqq Y^X(t)$ provided that the driver function $g\colon \O\times[0,T]\times\R^m\to\R$ is convex in the last variable and satisfies further regularity properties (e.g., Lipschitz in the last variable) to ensure the existence of a unique solution for the BSDE. Conversely, a given dynamic time consistent convex risk measure $\rho=(\rho_t)_{t\in[0,T]}$ for which $\rho_0$ satisfies two additional conditions (i.e., strict monotonicity and dominatedness, see \citet[Proposition 20]{G06} for details) can be represented as $\rho_t(X)=Y^X(t)$ through a family of BSDEs with a common driver $g$ as described above.

In contrast, no such representation is known for set-valued dynamic risk measures for random vectors. As a starting point, one can restrict attention to set-valued risk measures that can be written as the sum of a vector-valued function and the cone of positive random vectors, which is essentially equivalent to defining risk measures as vector-valued functions. In this case, as studied in \cite{X16}, a natural candidate that extends \eqref{riskBSDE} is a multidimensional BSDE, which has the same form as \eqref{riskBSDE} but with a vector-valued driver function $g$ and a vector-valued terminal value $X$. However, the only vector-valued functions that can be seen as risk measures (i.e., monotone, translative, convex) are the trivial ones whose coordinates are risk measures on individual coordinate spaces; this was observed in \cite{X16} in the dynamic setting via BSDE arguments and a more direct proof in the static setting was given recently in \cite{MMnonlinear}. Hence, multidimensional BSDEs have very limited use for the representation of dynamic risk measures in the multivariate framework. This motivates us to consider more general structures that we describe next.

In set-valued stochastic analysis, stochastic differential inclusions and set-valued stochastic differential equations appear as natural ways of extending stochastic differential equations to the set-valued framework. In the forward case, with the developments in set-valued stochastic integration, a theory of such inclusions and set-valued equations was developed in the last two decades; the reader is referred to the book \cite{kisielewicz} and the references therein for a detailed treatment. For set-valued risk measures, we consider the versions of these two structures that move backward in time:
\begin{itemize}
\item \textbf{Backward stochastic differential inclusions}: the vector-valued selectors of a set-valued process propagate backward in time with the path of the specific solution being of fundamental importance; this is a suitable choice if one is interested in characterizing a path moving inside the dynamic set-valued risk measure.
\item \textbf{Set-valued backward stochastic differential equations}: the entire set-valued process propagates backward in time with no specific selector itself deemed to be important; this is a more sophisticated structure as it characterizes the entire set-valued process that corresponds to the dynamic set-valued risk measure.
\end{itemize}
It is the understanding of these representations that motivate us in our current work.

Despite the considerable literature on forward inclusions and set-valued equations, at least in the case of driver functions with compact values, their backward counterparts are quite new in the literature. The few works dealing with backward stochastic differential inclusions are \cite{bsdi,bsdiweak}, which also assumed compact values for the driver. To the best of our knowledge, set-valued backward stochastic differential equations have not been studied in the literature yet. Hence, though we are motivated by the continuous time representations of risk measures described above, such differential inclusions and set-valued differential equations fall beyond the scope of this work.

In order to begin understanding the continuous time setting, we will investigate the discrete time approximation. As such, in this work, we consider backward stochastic difference inclusions (Section~\ref{sec:bsdi}) and set-valued backward stochastic difference equations (Section~\ref{sec:sv-bsde}) in order to understand fundamental questions about the representation of dynamic risk measures via a differential system, e.g., what mathematical form such representation may take and which approach (inclusion or set-valued equation) is more appropriate.  This permits us to both begin the study of continuous time multivariate risk measures and, at the same time, direct the theoretical research towards usable applications.

It should be noted that the fundamental BSDE representation in \eqref{riskBSDE}, which works under the filtration of a multidimensional Brownian motion, already covers a rich class of applications; see \citet[Section 3]{G06}, for instance. In our discrete time framework, we consider a multidimensional Bernoulli random walk as a standard discretization of Brownian motion. This is in line with the approach in \cite{stadje} for univariate risk measures. The multidimensional random walk models the underlying source of randomness and admits a predictable representation property in a generalized sense. This property is the key to having the representations via difference inclusions in Section \ref{sec:bsdi} and via set-valued difference equations in Section \ref{sec:sv-bsde}. In Remark \ref{generalpredrep}, we outline a more general setting where the Bernoulli random walk is replaced with a multidimensional square-integrable martingale in discrete time. In this case, a more abstract version of the predictable representation property is available and can be used in order to extend the results of this paper.

Beyond the immediate applications for risk measures, BSDEs are strongly related with the dynamic programming principle. The mean-variance and mean-risk problems are well-known to, generally, be time-inconsistent and do \emph{not} follow the dynamic programming principle as noted in \cite{timeinconsistency}. However, by considering the multi-objective formulation, \cite{KR20} has found that the dynamic programming principle in the mean-risk problem is satisfied in discrete time. Therefore, the insights we gain on expanding the BSDE representation of the dynamic programming principle to multivariate problems through either backward stochastic differential inclusions or set-valued backward stochastic differential equations is of wider interest and importance.

This work has two primary contributions we wish to highlight.  First, we present backward stochastic difference inclusion and set-valued backward stochastic difference equation representations for set-valued dynamic risk measures.  These are found in Sections~\ref{sec:bsdi} and~\ref{sec:sv-bsde} respectively.  Second, because neither inclusion nor set-valued equation representations have been studied in applications, we use set-valued risk measures in order to determine which generalization of BSDEs is a more promising avenue for considering the dynamic programming principle.  Specifically, as highlighted in this work, backward stochastic differential inclusions appear to be the proper methodology for future studies of the dynamic programming principle.

\section{Background}\label{sec:mptc}
Let $\mathbb{T}$ be some interval of time (discrete or continuous) with finite horizon $T = \sup \mathbb{T} > 0$. Let $(\O,\F,(\F_t)_{t \in \mathbb{T}},\Pr)$ be a filtered probability space satisfying the usual conditions and $d\geq 1$ an integer. Let $\G$ be a sub-$\sigma$-algebra of $\F$. We denote by $L_d^0(\G)$ the space of all $\G$-measurable random variables $X\colon\O\to\R^d$, where two elements are distinguished up to $\Pr$-almost sure equality. For each $p\in[1,+\infty)$, we denote by $L^p_d(\G)$ the set of all random variables $X\in L_d^0(\G)$ for which $\E\sqb{|X|^p}<+\infty$, where $|\cdot|$ is a fixed norm on $\R^d$. We denote by $L^\infty_d(\G)$ the set of all random variables $X\in L_d^0(\G)$ for which $|X|$ is $\Pr$-essentially bounded.  For ease of notation throughout the rest of this paper, we define $L_{d,+}^p(\G) := \cb{X \in L_d^p(\G) \; | \; \Pr(X \in \R^d_+) = 1}$. 

For $C,D\subseteq L_d^0(\F_T)$, the set $C+D\coloneqq \cb{X+Y \; | \; X\in C, Y\in D}$ is the Minkowski sum of $C,D$. If $Z\in L_1^0(\F_T)$, then $ZC\coloneqq CZ\coloneqq \cb{ZX \; | \; X\in C}$ is the set of pointwise products.

Let $p\in[1,+\infty]$. As defined in \cite{tc,mptc}, a conditional risk measure at time $t \in \mathbb{T}$ is a mapping
\[
R_t: L_d^p(\F_T) \to \P_+(L_d^p(\F_t)) \coloneqq \cb{C \subseteq L_d^p(\F_t) \; | \; C = C + L_{d,+}^p(\F_t)}
\]
of random vectors into upper sets.  As described above, we consider risk measures as functions of random vectors; the interpretation of these vectors depends on the context of the problem, i.e., markets with friction, systemic risk, etc.  (If the input can be aggregated into a single random variable ($d = 1$) without loss of information, then univariate risk measures (\cite{R04}) can be applied instead.) Throughout this work, we will focus solely on normalized convex risk measures, i.e., those that are:
\begin{enumerate}
\item normalized: $R_t(0) \neq \emptyset$, $\Pr(\tilde R_t(0) = \R^d) = 0$ (where $\tilde R_t(0)$ is an $\F_t$-measurable random set such that $R_t(0) = \{m \in L_d^p(\F_t) \; | \; \Pr(m \in \tilde R_t(0)) = 1\}$ and which exists due to the properties below and \citet[Theorem 2.1.10]{randomset}), and $R_t(X) = R_t(0) + R_t(X)$ for any $X \in L_d^p(\F_T)$;
\item translative: $R_t(X + m) = R_t(X) - m$ for any $X \in L_d^p(\F_T)$ and $m \in L_d^p(\F_t)$;
\item monotone: $R_t(X) \subseteq R_t(Y)$ for any $X,Y \in L_d^p(\F_T)$ such that $Y - X \in L_{d,+}^p(\F_T)$; 
\item conditionally convex: it holds
\[R_t(\lambda X + (1-\lambda)Y) \supseteq \lambda R_t(X) + (1-\lambda) R_t(Y)\]
for any $X,Y \in L_d^p(\F_T)$ and $\lambda \in L_1^p(\F_t)$ with $\Pr(\lambda \in [0,1]) = 1$; and
\item closed: the set $\operatorname{graph} R_t := \cb{(X,m) \in L_d^p(\F_T) \times L_d^p(\F_t) \; | \; m \in R_t(X)}$ is closed in the product topology and, as a consequence, $R_t(X)$ is closed for every $X \in L_d^p(\F_T)$.
\end{enumerate} 
A dynamic risk measure $R = (R_t)_{t \in \mathbb{T}}$ is a sequence of conditional risk measures over time.

Let $R=(R_t)_{t\in\mathbb{T}}$ be a dynamic risk measure. By construction (see, e.g.,~\citet[Proposition 2.8]{tc} and \citet[Lemma 3.6]{comparison}), $R$ is decomposable, i.e.,
\[
R_{t}(\ind_A X+ \ind_{A^c} Y) = \ind_A R_t(X)+\ind_{A^c} R_t(Y)
\]
for every $X,Y\in L_d^p(\F_{T}), A\in\F_t, t\in\mathbb{T}$, where $\ind_A$ denotes the probabilistic indicator function of event $A$.

\begin{remark}\label{rem:utility}
Though we present all results in this paper for set-valued risk measures, these results apply equally for set-valued utility functionals $U_t:  L_d^p(\F_T) \ni X\mapsto -R_t(X)$. In \cite{hamel2017set}, a (nontranslative) set-valued utility functional is defined based on expected utility and the associated utility maximization problem is studied. An axiomatic treatment of (translative) set-valued utility functionals defined on a set of closed convex random sets is given in \cite{MMnonlinear} under the name \emph{superlinear set-valued expectations}; coherent set-valued risk measures for random vectors can be embedded into this framework. As our constructions rely on some earlier results on set-valued risk measures, e.g., \cite{tc,comparison}, we consider these functionals in the current work for consistency with the most relevant literature.
\end{remark}

With the introduction of dynamic risk measures, the manner in which risks change over time is fundamentally important.  As studied in~\cite{tc,algorithm,bentaharlepinette,chenhu}, multiportfolio time consistency provides the appropriate definition from a mathematical and computational perspective.  A dynamic risk measure $R = (R_t)_{t \in \mathbb{T}}$ is said to be multiportfolio time consistent if, for any $X \in L_d^p(\F_T)$ and $\mathbb{Y} \subseteq L_d^p(\F_T)$,
\[
R_s(X) \subseteq R_s[\mathbb{Y}] := \bigcup_{Y \in \mathbb{Y}} R_s(Y) \; \Rightarrow \; R_t(X) \subseteq R_t[\mathbb{Y}]
\]
for all times $t,s \in \mathbb{T}$ with $t < s$. (To clarify notation, functions $R_t(\cdot)$ are applied to random vectors whereas $R_t[\cdot]$ are applied to sets of random vectors as the union over each element as defined above.) Multiportfolio time consistency states that if every risk compensating portfolio for $X$ at time $s$ also compensates for the risk of some portfolio $Y \in \mathbb{Y}$, then that same relation holds for any time $t < s$.  That is, the risk of \emph{baskets} of portfolios is consistent over time. This version of time consistency is related to the dynamic programming principle through a recursive relation
\begin{equation}\label{eq:mptc-recursive}
R_t(X) = R_{t,s}[-R_s(X)]
\end{equation}
for every $X \in L_d^p(\F_T)$, $t,s \in \mathbb{T}$ with $t < s$. This recursive formulation is constructed via the stepped risk measure $R_{t,s}: L_d^p(\F_s) \to \P_+(L_d^p(\F_t))$, which is defined as the restriction of the conditional risk measure $R_t$ to the domain $L_d^p(\F_s) \subseteq L_d^p(\F_T)$, i.e.,
\[
R_{t,s}(Z) = R_t(Z)
\]
for every $Z \in L_d^p(\F_s)$.

To consider local versions of the difference equations and inclusions, it is preferable to consider the random set approach to conditional risk measures.  By the equivalence between closed decomposable sets of random vectors and closed random sets (see, e.g., \citet[Theorem 2.1.10]{randomset}) as applied on set-valued risk measures in \cite{comparison}, there exists a random set $\tilde R_t(X)$ (i.e., $\tilde R_t(X)(\omega) \subseteq \R^d$ for every $\omega \in \O$) such that
\[
R_t(X) = \cb{m \in L_d^p(\F_t) \; | \; \Pr(m \in \tilde R_t(X)) = 1}
\]
for every $X \in L_d^p(\F_T), t \in \mathbb{T}$.
As the rest of the work is concerned solely with a finitely generated filtration, we will abuse notation to define $R_t(X)(\omega) := \tilde R_t(X)(\omega)$ for every $X \in L_d^p(\F_T), t \in \mathbb{T}, \omega \in \O$. 
These definitions are given likewise for the stepped risk measures $R_{t,s}$ with $t,s \in \mathbb{T}$ and $t < s$. 

\section{Discrete Time Setting}\label{sec:discrete}

In this section, we present the mathematical notation and setting for the discrete time stochastic processes of interest. In particular, we introduce an $m$-dimensional Bernoulli random walk and consider the predictable representation property with respect to this process. We present these results for completeness and to assist in notation later in this paper; more details can be found in \cite{protter}. We wish to highlight the need for the orthogonal martingale term when considering an $m > 1$ dimensional Bernoulli random walk which complicates the setting from, e.g.,~\cite{stadje}.  Such considerations are similarly required if a different random walk was chosen as the driving process; see Remark~\ref{generalpredrep} for more details of this consideration. We focus on the Bernoulli random walk both because it presents a simple mathematical setting that allows for an explicit orthogonal process and because of its ubiquity as a discrete time representation that limits to the Brownian motion.

Let $T>0$ be a fixed time. Let $m,K\geq 1$ be integers and consider a time set $\mathbb{T}=\cb{t_0,\ldots,t_K}$, where $0=t_0<t_1<\ldots<t_{K-1}<t_K=T$. Let us write $\Delta t_k\coloneqq t_{k}-t_{k-1}$ for each $k\in\cb{1,\ldots,K}$. Let $(\O,\F,\Pr)$ be a probability space on which there exist independent $m$-dimensional random vectors $B(1),\ldots,B(K)\colon\O\to\cb{-1,+1}^m$, where the components of $B(k)=(B_{1}(k),\ldots,B_{m}(k))^\T$ are possibly correlated symmetric Rademacher random variables for each $k\in\cb{1,\ldots,K}$. We define an $m$-dimensional random walk $M=(M_1(t),\ldots,M_m(t))_{t\in\mathbb{T}}$ by $M_i(t_0)\coloneqq 0$ and
\begin{equation}\label{randomwalkdefn}
M_i(t_k) \coloneqq M_i(t_{k-1}) + \sqrt{\Delta t_{k}} B_i(k)
\end{equation}
for each $k\in\cb{1,\ldots,K}$ and $i\in\cb{1,\ldots,m}$.  For notational simplicity, let
\[
\Delta M(t_k) \coloneqq M(t_k) - M(t_{k-1}) = \sqrt{\Delta t_k}B(k).
\]
Let $(\F^M_{t})_{t\in\mathbb{T}}$ be the natural filtration of $M$, that is, $\F^M_{t_0}=\cb{\emptyset,\Omega}$ and
\[
\F^M_{t_k}=\sigma(M(t_1),\ldots,M(t_k))=\sigma(B(1),\ldots,B(k))
\]
for each $k\in\cb{1,\ldots,K}$.

Let $k\in\cb{1,\ldots,K}$. Since $\F^M_{t_k}$ is generated by finitely many events, every real-valued random variable that is measurable with respect to $\F^M_{t_k}$ is bounded. Hence, $L_d^0(\F^M_{t_k})=L^p_d(\F^M_{t_k})$ for each $p\in[1,+\infty]$. 
To simplify notation for the remainder of this work, we define $\lcal_{t_k}^d := L_d^0(\F^M_{t_k})$.
When $m=1$, every $X\in \lcal_{t_k}^d$ can be written as
\begin{equation}\label{predrep}
X=\xi+\psi \Delta M(t_{k})
\end{equation}
for some $\xi,\psi\in \lcal_{t_{k-1}}^d$ (see, e.g., \cite{fs:sf}). This is called the \emph{predictable representation property} of the one-dimensional random walk $M$. For $m\geq 2$, the analogous representation
\begin{equation}\label{predrepwrong}
\xi+\sum_{i=1}^m\psi_i \Delta M_i(t_k)
\end{equation}
with $\xi,\psi_1,\ldots,\psi_m\in \lcal_{t_{k-1}}^d$ \emph{fails} to equal $X\in \lcal_{t_k}^d$, in general.

Instead of \eqref{predrepwrong}, a proper generalization of \eqref{predrep} should take into account the ``cross-terms" created by the $m$ components of the vector-valued random walk $M$ (i.e., $M_1,\ldots,M_m$). To that end, let us denote by $\mathcal{I}$ the set of all nonempty subsets of $\cb{1,\ldots,m}$. For each $ I\in\mathcal{I}$, we define a process $M_I=(M_I(t))_{t\in\mathbb{T}}$ by $M_I(t_0)\coloneqq 0$ and
\[
\Delta M_I(t_k)\coloneqq M_I(t_k)-M_I(t_{k-1})=\sqrt{\Delta t_k}B_I(k),\; \text{where}\; B_I(k)\coloneqq \bigtimes_{i\in I}B_i(k)
\]
for each $k\in\cb{1,\ldots,K}$, by a slight abuse of notation. In addition to the earlier processes $M_{\cb{i}}=M_i$, $i\in\cb{1,\ldots,m}$, this definition creates $2^m-m-1$ new processes. Moreover, when the components of $B(k)$ are independent for each $k\in\cb{1,\ldots,K}$, all the $2^m-1$ processes are martingales that are orthogonal. The new form of predictable representation is stated in the next lemma.

\begin{lemma}\label{genpredrep}
	Let $k\in\cb{1,\ldots,K}$. Every $Y\in \lcal_{t_k}^d$ can be written as
	\[
	Y=\xi+\sum_{I\in\mathcal{I}}\psi_{I}\Delta M_I(t_k)
	\]
	for some $\xi,\psi_I\in \lcal_{t_{k-1}}^d$, $I\in\mathcal{I}$.
	\end{lemma}

\begin{remark}
	For $d=1$, the result appears as Lemma~6.1 of a preprint version of \cite{cheriditohorst}; its proof there is based on a spanning argument for the finite-dimensional vector space $\lcal_{t_k}^1$ assuming $k=1$. We provide a more elementary and complete proof below with an explicit derivation of the predictable representation. Unlike the proof in the preprint version of \cite{cheriditohorst}, our proof is also valid without assuming that the components of the random walk are independent.
	\end{remark}

\begin{proof}[Proof of Lemma~\ref{genpredrep}]
	Let us assume that $d=1$. As a first step, let $Y\in L_1^0(\sigma(B(k)))$. Note that $\sigma(B(k))$ is generated by $2^m$ events of the form $\cb{B(k)=b}$ with $b=(b_1,\ldots,b_m)\in \cb{-1,+1}^m$ that partition $\O$. Hence, we may write
\[
Y =\sum_{b\in\cb{-1,+1}^m}c_{b}\ind_{\cb{B(k)=b}}=  \sum_{b\in\cb{-1,+1}^m}c_{b}\prod_{i=1}^m \ind_{\cb{B_i(k)=b_i}}
\]
for some $c_b\in\R$, $b\in\cb{-1,+1}^m$. Note that $\ind_{\cb{B_i(k)=b_i}}=\frac{1+b_i B_i(k)}{2}$ for each $i\in\cb{1,\ldots,m}$ and $b\in\cb{-1,+1}^m$. Hence,
\[
Y= \sum_{b\in\cb{-1,+1}^m}c_{b}\prod_{i=1}^m \frac{1+b_i B_i(k)}{2}.
\]
One can rewrite the above sum in the form of a polynomial of $B_1(k),\ldots,B_m(k)$. The constant term of this polynomial is
\[
\xi\coloneqq \frac{1}{2^m}\sum_{b\in\cb{-1,+1}^m}c_b.
\]
On the other hand, each functional term of the polynomial is of the form $\tilde{\psi}_I B_{I}(k)$, where $I\in\mathcal{I}$ and
\[
{\tilde{\psi}}_I \coloneqq \frac{1}{2^m}\sum_{b\in\cb{-1,+1}^m} c_b\prod_{i\in I}b_i.
\]
Taking $\psi_I\coloneqq \frac{\tilde{\psi}_I}{\sqrt{\Delta t_k}}$ for each $I\in\mathcal{I}$, we obtain
\[
Y=\xi+\sum_{I\in\mathcal{I}}\tilde{\psi}_{I}B_I(k)=\xi+\sum_{I\in\mathcal{I}}\psi_I \Delta M_I(t_k).
\]

If $k=1$, then the conclusion of the lemma follows by the first step since $\F_{t_1}^M=\sigma(B(1))$. Next, suppose that $k\in\cb{2,\ldots,K}$ and let $Y\in \lcal_{t_k}^1$. Hence, we may write
\[
Y=f\circ(B(1),\ldots,B(k))
\]
for some (Borel measurable) function $f\colon\cb{-1,+1}^{km}\to \R$. In particular,
\begin{equation}\label{Xrep}
Y=\sum_{b(1)\in\cb{-1,+1}^m}\ldots\sum_{b(k-1)\in\cb{-1,+1}^m}f\circ(b(1),\ldots,b(k-1),B(k))\ind_{\cb{B(1)=b(1),\ldots,B(k-1)=b(k-1)}}.
\end{equation}
Let $b(1),\ldots,b(k-1)\in \cb{-1,+1}^m$. Since $f\circ (b(1),\ldots,b(k-1),B(k))\in L_1^0(\sigma(B(k)))$, by the first step, there exist $g(b(1),\ldots,b(k-1))\in\R$, $h_I(b(1),\ldots,b(k-1))\in\R$ for each $I\in \mathcal{I}$ such that
\begin{equation}\label{fgh}
f\circ (b(1),\ldots,b(k-1),B(k))=g(b(1),\ldots,b(k-1))+\sum_{I\in\mathcal{I}}h_I(b(1),\ldots,b(k-1))\Delta M_I(t_k).
\end{equation}
Then, \eqref{Xrep} and \eqref{fgh} imply that
\[
Y = g\circ (B(1),\ldots,B(k-1))+\sum_{I\in\mathcal{I}}h_I\circ(B(1),\ldots,B(k-1))\Delta M_I(t_k).
\]
The functions $g,h_I\colon \cb{-1,+1}^{(k-1)m}\to\R$, $I\in\mathcal{I}$, are Borel measurable since they are defined on finite sets. Hence, by taking
\[
\xi\coloneqq g\circ (B(1),\ldots,B(k-1))\in \lcal_{t_{k-1}}^1, \quad \psi_I\coloneqq h_I\circ (B(1),\ldots,B(k-1))\in \lcal_{t_{k-1}}^1,\ I\in\mathcal{I},
\]
the conclusion of the lemma follows for $d=1$.

For arbitrary $d\geq 1$, applying the result for $d=1$ to each component of $Y\in \lcal_{t_k}^d$ gives the claimed representation.
	\end{proof}

\begin{remark}\label{generalpredrep}
	One can consider a more general setting in which $M$ is a square-integrable martingale such as a symmetric random walk with an arbitrary distribution for its increments, e.g., $M$ can be a symmetric Gaussian random walk. In this case, the simple predictable representation in \eqref{predrepwrong} fails to hold in general even for the case $m=1$. Nevertheless, when the $m$ components of $M$ are independent, by the results in \citet[Section IV.3]{protter}, one has the more general form
	\begin{equation}\label{protterpredrep}
	X = \xi + \sum_{i=1}^m\psi_i\Delta M_i(t_k) + \Delta N(t_k),
	\end{equation}
	where $\xi,\psi_1,\ldots,\psi_m\in\lcal_{t_{k-1}}^d$, $\Delta N(t_k)=N(t_k)-N(t_{k-1})$ and $N=(N_1(t),\ldots,N_d(t))_{t\in\mathbb{T}}$ is a square-integrable martingale that is orthogonal to $M$, i.e., $\E\sqb{\Delta M_i(t_k)\Delta N_j(t_k)}=0$ for every $i\in\cb{1,\ldots,m}$, $j\in\cb{1,\ldots,d}$, $k\in\cb{1,\ldots,K}$. The representation in Lemma \ref{genpredrep} above can be seen as a special case of \eqref{protterpredrep} for the Bernoulli random walk but it also provides the explicit structure of the orthogonal term as
	\begin{equation}\label{exporthogonal}
	\Delta N(t_k)=\sum_{I\in\mathcal{I}\backslash\cb{\cb{1},\ldots,\cb{m}}}\psi_I\Delta M_I(t_k).
	\end{equation}
	When $M$ is a general square-integrable martingale, \citet[Proposition 2.1]{jacod} has an explicit formula for the integrands $\psi_1,\ldots,\psi_m$; however, the structure in \eqref{exporthogonal} for the orthogonal term is very specific to the Bernoulli case. (Similarly, for a Gaussian random walk, the orthogonal term in \eqref{protterpredrep} is referred to as the \emph{martingale representation error} in \cite{discreteocone} and is made explicit via a discrete version of Clark-Ocone formula \citep[Theorem 2.1]{discreteocone}.) The Bernoulli random walk makes it possible to express the sum of the stochastic integral (with respect to $M$) and the orthogonal terms in \eqref{protterpredrep} compactly as $\sum_{I\in\mathcal{I}}\psi_I\Delta M_I(t_k)$, which simplifies notation in the backward stochastic difference inclusions and set-valued backward stochastic difference equations to follow. On the other hand, when passing to continuous time, we may use a sequence of Bernoulli random walks in order to approximate Brownian motion, which is the underlying process of the fundamental BSDE representations in continuous time (see \eqref{riskBSDE}). For these reasons, we prefer to work within the Bernoulli framework and use the predictable representation in Lemma \ref{genpredrep}. 
	\end{remark}

We conclude this section with an illustrative example to be developed further in the later sections.

\begin{example}\label{exmodel1}
	We consider a simple market model with proportional transaction costs. For notational simplicity, we assume that $t_k=k$ so that $\Delta t_k=1$ for each $k\in\cb{1,\ldots,K}$. The market is driven by $m=d-1$ symmetric and correlated random walks, and it has $d$ assets. Asset $0$ is a riskless bond with no transaction costs; its (deterministic) price process $(A_t)_{t\in\mathbb{T}}$ is given by $A_t=A_0(1+r)^t$, $t\in\cb{1,\ldots,K}$, where $A_0>0$ is the initial bond price and $r>-1$ is the periodic interest rate. To simplify the presentation, let us assume that $A_0=1$ and $r=0$ so that $A_t=1$ for each $t\in\cb{1,\ldots,K}$. In addition, for each $i\in\cb{1,\ldots,d-1}$, asset $i$ is a risky stock with proportional transaction costs; its mid-market price process $(S^i_t)_{t\in\mathbb{T}}$ (denoted in terms of the num\'{e}raire of the bond) follows a Cox-Ross-Rubinstein model given by
	\[
	S^i_t = S^i_0e^{\sigma_i M_i(k)},\quad t\in\cb{1,\ldots,K},
	\]
	where $S^i_0>0$ is the initial mid-market price and $\sigma_i>0$ is the volatility coefficient. Furthermore, there is a fixed transaction cost rate $\lambda_i\geq 0$ for the stock, which is the same for all periods for simplicity. Hence, the bid and ask prices at time $t\in\mathbb{T}$ are given by
	\[
	S_t^{i,b}=S^i_t(1-\lambda_i),\quad S_t^{i,a} = S^i_t(1+\lambda_i),
	\]
	respectively. Let us define the price (row) vectors
	\[
	S_t^b\coloneqq (S_t^{1,b},\ldots,S_t^{d-1,b}),\quad  S_t^a\coloneqq (S_t^{1,a},\ldots,S_t^{d-1,a}).
	\]
	
	For simplicity, we assume that all transactions are performed via the bond, that is, a direct conversion between any two risky assets is not allowed. For each $t\in\mathbb{T}$ and $\o\in\O$, we denote by $K_t(\o)$ the solvency cone of all $d$-dimensional portfolio vectors that can be exchanged into positive portfolio vectors. In our case, $K_t(\o)$ is the convex cone generated by the columns of the matrix (see \citet[Section 7]{algorithm} and \citet[Section 5]{LR14})
	\begin{equation}\label{columns}
	\begin{pmatrix}-S_t^b(\o) & S_t^a(\o)\\ I_{d-1}& -I_{d-1}\end{pmatrix},
	\end{equation}
	where $I_{d-1}$ is the $(d-1)\times(d-1)$ identity matrix. From the structure of the Cox-Ross-Rubinstein model, it is clear that $K_t(\o)$ depends on $\o$ only through $M(t)(\o)$, the realization of the underlying $m$-dimensional random walk. Hence, by a slight abuse of notation, we write
	\begin{equation}\label{Ktnotation}
	K_t(\o)=K_t(c)
	\end{equation}
	whenever $M(t)(\o)=c$.
	
	When $d=2$, we write $\sigma=\sigma_1, S_t^{b}=S_t^{1,b}, S_t^a=S_t^{1,a}$ and, from \eqref{columns}, we have
	\[
	K_t(\o)=\co\cb{\begin{pmatrix}-S_t^{b}(\o)\\1\end{pmatrix},\begin{pmatrix} S_t^{a}(\o)\\ -1\end{pmatrix}},
	\]
where $\co$ denotes the convex conic hull operator; when $d=3$, we have
	\[
	K_t(\o)=\co\cb{\begin{pmatrix}-S_t^{1,b}(\o) \\ 1\\ 0\end{pmatrix},\begin{pmatrix}-S_t^{2,b}(\o) \\ 0 \\ 1\end{pmatrix},\begin{pmatrix} S_t^{1,a}(\o)\\ -1 \\ 0\end{pmatrix},\begin{pmatrix} S_t^{2,a}(\o)\\ 0\\ -1\end{pmatrix}}.
	\]

For each $t\in\mathbb{T}$, let us define $L_d^0(\F^M_t,K_t)=\{Y\in\lcal_{t}^d \; | \; \Pr\{\o\in\O\mid Y(\o)\in K_t(\o)\}=1\}$, the set of all measurable selectors of the random set $K_t$. Let $C_{t,T}=-\sum_{s=t}^T L_d^0(\F^M_s,K_s)$, which is the set of all self-financing portfolios in this market. Then, for a multidimensional claim $X\in\lcal_t^d$, the set
\[
SHP_t(X)\coloneqq \cb{Y\in\lcal_t^d \; | \; X\in Y+C_{t,T}}
\]
is the set of all superhedging portfolios of $X$ at time $t\in\mathbb{T}$. Letting
\[
R^{SHP}_t(X)=SHP_t(-X),
\]
the family $R^{SHP}=(R^{SHP}_t)_{t\in\mathbb{T}}$ is a coherent multiportfolio time consistent dynamic risk measure as studied in \citet[Example 5.4]{mptc}.
	\end{example}

\section{Backward Stochastic Difference Inclusion}\label{sec:bsdi}

In this section, we show that a given dynamic set-valued risk measure in discrete time gives rise to a backward stochastic difference inclusion (BS$\Delta$I). Since we do not consider the continuous time limits of these risk measures, we derive a BS$\Delta$I without scaling and tilting the original risk measure as was done in \cite{stadje}. BS$\Delta$Is present a recursive formulation for the computation of the set-valued risk measure as a collection of selectors.  As such, the proposed BS$\Delta$I provides a methodology for studying singular capital allocation strategies over time to manage risk. Because a risk manager would only ever implement a single strategy, such a construction provides exactly the dynamics an investor would need to consider.  Therefore, in this section, we are interested in the backward stochastic difference inclusion that encodes the dynamic programming principle of multiportfolio time consistency.

In the setting of Section~\ref{sec:discrete}, we work with the filtered probability space $(\O,\F, (\F^M_{t})_{t\in\mathbb{T}}, \Pr)$, where $M$ is the $m$-dimensional random walk defined in \eqref{randomwalkdefn}. Let us consider a multiportfolio time consistent dynamic set-valued convex risk measure $R=(R_{t})_{t\in\mathbb{T}}$ with one-step conditional risk measures $R_{t_{k-1},t_{k}}\colon \lcal_{t_{k}}^d\to \P_+(\lcal_{t_{k-1}}^d)$, $k\in\cb{1,\ldots,K}$. For the terminal risk measure $R_{t_K}\colon \lcal_{t_K}^d\to \P_+(\lcal_{t_K}^d)$, note that we have $R_{t_K}(X)=-X+R_{t_K}(0)$ for each $X\in \lcal_{t_K}^d$.

We first relate $R$ to a BS$\Delta$I with a \emph{nonlocal} driver. To that end, let us introduce the domain
\begin{equation}\label{dominc}
\mathbb{D} \coloneqq \cb{(t,\psi) \; | \; t \in \cb{t_0,\ldots,t_{K-1}}, \; \psi=(\psi_I)_{I\in\mathcal{I}} \in (\lcal_{t}^d)^{\mathcal{I}}}
\end{equation}
and define the set-valued driver $G\colon\mathbb{D}\to 2^{\lcal^d}$ by
\begin{equation}\label{nonlocaldefn}
G\of{t_{k-1},\psi} \coloneqq  \frac{1}{\Delta t_{k}}R_{t_{k-1},t_k}\of{-\sum_{I\in\mathcal{I}} \psi_I\Delta M_I(t_k)}
\end{equation}
for each $(t_{k-1},\psi)\in\mathbb{D}$ with $k\in\cb{1,\ldots,K}$. Note that $G$ is adapted, i.e., $G(t,\psi)\in \P_+(\lcal_t^d)$ whenever $(t,\psi)\in\mathbb{D}$.

\begin{proposition}\label{nonlocalprop}
	Let $k\in\cb{1,\ldots,K}$ and $Y(t_k)\in \lcal_{t_k}^d$. Consider the one-step BS$\Delta$I
	\begin{equation}\label{1step}
	Y(t_{k-1})\in Y(t_k)+G(t_{k-1},\psi(t_{k-1}))\Delta t_k-\sum_{I\in\mathcal{I}}\psi_I(t_{k-1})\Delta M_I(t_k)
	\end{equation}
    for some $\psi(t_{k-1})\in (\lcal_{t_{k-1}}^d)^{\mathcal{I}}$.
	Then, the set $R_{t_{k-1},t_{k}}(-Y(t_k))$ coincides with the reachable set of \eqref{1step}, that is,
	\begin{equation}\label{reachable1step}
	R_{t_{k-1},t_k}(-Y(t_k))=\cb{Y(t_{k-1})\in\lcal_{t_{k-1}}^d \; | \; \text{\eqref{1step} holds for some }\psi(t_{k-1})\in(\lcal_{t_{k-1}}^d)^{\mathcal{I}}}.
	\end{equation}
\end{proposition}
\begin{proof}
	Let $Y(t_{k-1})\in R_{t_{k-1},t_{k}}(-Y(t_{k}))$. By Lemma~\ref{genpredrep}, $Y(t_k)$ has the predictable representation
	\begin{equation}\label{predrep2}
	Y(t_k) = \xi(t_{k-1})+\sum_{I\in\mathcal{I}} \psi_I(t_{k-1})\Delta M_I(t_k)
	\end{equation}
	for some $\xi(t_{k-1}), \psi_I(t_{k-1})\in \lcal_{t_{k-1}}^d, I\in\mathcal{I}$. By \eqref{nonlocaldefn}, we have
	\begin{align*}
	&Y(t_{k-1})-Y(t_{k})\\
	&\in R_{t_{k-1},t_{k}}(-Y(t_k))-Y(t_k)\\
	&= R_{t_{k-1},t_k}\of{-\xi(t_{k-1}) -\sum_{I\in\mathcal{I}} \psi_I(t_{k-1})\Delta M_I(t_k)}-\xi(t_{k-1}) -\sum_{I\in\mathcal{I}} \psi_I(t_{k-1})\Delta M_I(t_k)\\
	&= R_{t_{k-1},t_k}\of{-\sum_{I\in\mathcal{I}} \psi_I(t_{k-1})\Delta M_I(t_k)}-\sum_{I\in\mathcal{I}} \psi_I(t_{k-1})\Delta M_I(t_k)\\
	&= G\of{t_{k-1},\psi(t_{k-1})}\Delta t_{k}-\sum_{I\in\mathcal{I}} \psi_I(t_{k-1})\Delta M_I(t_k).
	\end{align*}
	Hence, \eqref{1step} holds.

	Conversely, let $Y(t_{k-1})\in\lcal_{t_{k-1}}^d$ be such that the BS$\Delta$I \eqref{1step} holds for some $\psi(t_{k-1})\in (\lcal_{t_{k-1}}^d)^{\mathcal{I}}$. By the BS$\Delta$I, there exists $V(t_{k-1})\in G(t_{k-1},\psi(t_{k-1}))$ such that
	\[
	Y(t_{k-1}) = Y(t_k) +  V(t_{k-1})\Delta t_{k}-\sum_{I\in\mathcal{I}}\psi_I(t_{k-1})\Delta M_I(t_k),
	\]
	that is, $Y(t_k)$ has the predictable representation
	\[
	Y(t_{k}) = \xi(t_{k-1})+\sum_{I\in\mathcal{I}}\psi_I(t_{k-1})\Delta M_I(t_k),
	\]
	where $\xi(t_{k-1})\coloneqq Y(t_{k-1}) -V(t_{k-1})\Delta t_{k}\in\lcal_{t_{k-1}}^d$. Recalling the definition of the driver in \eqref{nonlocaldefn}, the BS$\Delta$I yields
	\begin{align*}
	Y(t_{k-1})&\in Y(t_k)+R_{t_{k-1},t_k}\of{-\sum_{I\in\mathcal{I}} \psi_I(t_{k-1})\Delta M_I(t_k)}-\sum_{I\in\mathcal{I}}\psi_I(t_{k-1})\Delta M_I(t_k)\\
	&=Y(t_k)+R_{t_{k-1},t_k}\of{-Y(t_{k}) +\xi(t_{k-1})}-\sum_{I\in\mathcal{I}}\psi_I(t_{k-1})\Delta M_I(t_k)\\
	&=Y(t_k)+R_{t_{k-1},t_k}\of{-Y(t_{k}) }- \xi(t_{k-1})-\sum_{I\in\mathcal{I}}\psi_I(t_{k-1})\Delta M_I(t_k)\\
	&=Y(t_k)+R_{t_{k-1},t_k}\of{-Y(t_{k}) }- Y(t_k)=R_{t_{k-1},t_k}\of{-Y(t_{k}) },
	\end{align*}
	which completes the proof.
	\end{proof}

\begin{corollary}\label{nonlocalcor}
Let $X\in \lcal_{t_K}^d$. If $(Y(t))_{t\in\mathbb{T}}$ is a process such that $Y(t_K)\in R_{t_K}(X)$ and $Y(t_{k-1})\in R_{t_{k-1},t_{k}}(-Y(t_{k}))$ for each $k\in\cb{1,\ldots,K}$, then there exists $\psi(t_{k-1})\in (\lcal_{t_{k-1}}^d)^{\mathcal{I}}$ for each $k\in\cb{1,\ldots,K}$ such that the BS$\Delta$I
	\begin{align*}
	&Y(t_{k-1})\in Y(t_k)+G(t_{k-1},\psi(t_{k-1}))\Delta t_k-\sum_{I\in\mathcal{I}}\psi_I(t_{k-1})\Delta M_I(t_k),\quad k\in\cb{1,\ldots,K},\\
	& Y(t_K)\in -X+R_{t_K}(0)
	\end{align*}
	holds. Conversely, if there exist a $d$-dimensional adapted process $(Y(t))_{t\in\mathbb{T}}$ and $\psi(t_{k-1})\in (\lcal_{t_{k-1}}^d)^{\mathcal{I}}$ for each $k\in\cb{1,\ldots,K}$ such that the above BS$\Delta$I holds, then $Y(t_K)\in R_{t_K}(X)$ and $Y(t_{k-1})\in R_{t_{k-1},t_k}(-Y(t_k))$ for each $k\in\cb{1,\ldots,K}$. In each case, the multi-step version of the BS$\Delta$I
\begin{align*}
&Y(t_k)\in Y(t_K)+\sum_{\ell=k+1}^K G(t_{\ell-1},\psi(t_{\ell-1}))\Delta t_\ell-\sum_{\ell=k+1}^K\sum_{I\in\mathcal{I}}\psi_I(t_{\ell-1})\Delta M_I(t_\ell),\quad k\in\cb{0,\ldots,K-1},\\
& Y(t_K)\in -X+R_{t_K}(0)
\end{align*}
holds as well.
\end{corollary}

\begin{proof}
By translativity, $R_{t_K}(X)=-X+R_{t_K}(0)$. From this and Proposition~\ref{nonlocalprop}, the two claims about the one-step BS$\Delta$I follow immediately. The claim about the multi-step BS$\Delta$I follows by iterating the one-step version. Indeed, for each $k\in\cb{0,\ldots,K-1}$, we have
\begin{align*}
Y(t_k) &= Y(t_K) +\sum_{\ell=k+1}^{K}(Y(t_{\ell-1})-Y(t_{\ell}))\\
&\in Y(t_K) + \sum_{\ell=k+1}^{K}G\of{t_{\ell-1},(\psi_I(t_{\ell-1}))_{I\in\mathcal{I}}}\Delta t_{\ell}-\sum_{\ell=k+1}^{K} \sum_{I\in\mathcal{I}}\psi_I(t_{\ell-1})\Delta M_I(t_{\ell}),
\end{align*}
as desired.
\end{proof}

Note that Corollary~\ref{nonlocalcor} shows that each process $(Y(t))_{t\in\mathbb{T}}$ that is a solution of the one-step BS$\Delta$I is a ``path'' in the dynamic risk measure evaluated at $X$, and vice versa. Moreover, such processes are also solutions of the multi-step BS$\Delta$I. The next corollary provides a partial converse to the latter statement.

\begin{corollary}\label{multistepcor}
	Let $X\in\lcal_{t_K}^d$ and $k\in\cb{0,\ldots,K}$. Consider the multi-step BS$\Delta$I
	\begin{align}
	&Y(t_k)\in Y(t_K)+\sum_{\ell=k+1}^K G(t_{\ell-1},\psi(t_{\ell-1}))\Delta t_\ell-\sum_{\ell=k+1}^K\sum_{I\in\mathcal{I}}\psi_I(t_{\ell-1})\Delta M_I(t_\ell).\label{multistepinc}
	\end{align}
	Then, the set $R_{t_k}(X)$ coincides with the reachable set of \eqref{multistepinc}, that is,
	\begin{equation}\label{reachableclaim}
	R_{t_k}(X)= \left\{Y(t_k) \in \lcal_{t_k}^d \; \left| \; \begin{array}{l}\text{\eqref{multistepinc} holds for some } Y(t_K)\in -X+R_{t_K}(0),\\ \psi(t_{k})\in(\lcal_{t_k}^d)^{\mathcal{I}},\ldots,\psi(t_{K-1})\in(\lcal_{t_{K-1}}^d)^{\mathcal{I}} \end{array}\right.\right\}.
	\end{equation}
\end{corollary}
\begin{proof}
	For each $k\in\cb{0,\ldots,K-1}$, let us denote by $\bar{R}_{t_k}(X)$ the reachable set on the right-hand of \eqref{reachableclaim}. We will show that $\bar{R}_{t_k}(X)=R_{t_k}(X)$ for every $k\in\cb{0,\ldots,K-1}$ by backward induction on $k$. For the base case, we have $\bar{R}_{t_{K-1}}(X)=R_{t_{K-1}}(X)$ directly from Corollary~\ref{nonlocalcor}. For the inductive case, let $k\in\cb{1,\ldots,K-1}$ and assume that $\bar{R}_{t_k}(X)=R_{t_k}(X)$. Consider inclusion~\eqref{multistepinc} at time $t_{k-1}$, i.e., the inclusion
	\begin{align}
	&Y(t_{k-1})\in Y(t_K)+\sum_{\ell=k}^K G(t_{\ell-1},\psi(t_{\ell-1}))\Delta t_\ell-\sum_{\ell=k}^K\sum_{I\in\mathcal{I}}\psi_I(t_{\ell-1})\Delta M_I(t_\ell).\label{multistepinc-1}
	\end{align}
	 Then,
	\begin{align*}
	&\bar{R}_{t_{k-1}}(X)\\
	&= \left\{Y(t_{k-1}) \in \lcal_{t_{k-1}}^d \; \left| \; \begin{array}{l}\text{\eqref{multistepinc-1} holds for some } Y(t_K)\in -X+R_{t_K}(0),\\ \psi(t_{k-1})\in(\lcal_{t_{k-1}}^d)^{\mathcal{I}},\psi(t_{k})\in(\lcal_{t_{k}}^d)^{\mathcal{I}},\ldots,\psi(t_{K-1})\in(\lcal_{t_{K-1}}^d)^{\mathcal{I}} \end{array}\right.\right\}\\
	&= \left\{Y(t_{k-1}) \in \lcal_{t_{k-1}}^d \; \left| \; \begin{array}{l} Y(t_{k-1}) \in Y(t_k) + G(t_{k-1},\psi(t_{k-1}))\Delta t_k - \sum_{I \in \mathcal{I}} \psi_I(t_{k-1})\Delta M_I(t_k) \\ \text{holds for some }Y(t_k) \in R_{t_k}(X), \; \psi(t_{k-1}) \in (\lcal_{t_{k-1}}^d)^{\mathcal{I}}\end{array}\right.\right\}\\
	&= \left\{Y(t_{k-1}) \in R_{t_{k-1},t_k}(-Y(t_k)) \; | \; Y(t_k) \in R_{t_k}(X)\right\}\\
	&= R_{t_{k-1},t_k}[-R_{t_k}(X)] = R_{t_{k-1}}(X).
	\end{align*}
	In this calculation, the second equality follows from applying the induction hypothesis; the third equality follows from Proposition~\ref{nonlocalprop}; and the last equality follows from multiportfolio time consistency.
	\end{proof}

The driver function $G$ defined in \eqref{nonlocaldefn} can be considered \emph{nonlocal} for the following two reasons: its second argument is a $d\times|\mathcal{I}|$-dimensional random vector (rather than a deterministic vector) and its output is a set of random vectors (rather than a deterministic set). We first aim to rewrite the BS$\Delta$I in Corollary~\ref{nonlocalcor} using a \emph{semi-local} driver $g$ whose second argument is a deterministic vector but the output is still a set of random vectors. Let us define $g\colon\cb{t_0,\ldots,t_{K-1}}\times (\R^d)^{\mathcal{I}}\to 2^{\lcal^d_{t_K}}$ by
\begin{equation}\label{semilocal}
g\of{t_{k-1},z} =  \frac{1}{\Delta t_{k}}R_{t_{k-1},t_k}\of{-\sum_{I\in\mathcal{I}} z_I\Delta M_I(t_k)}
\end{equation}
for each $k\in\cb{0,\ldots,K-1}, z=(z_I)_{I\in\mathcal{I}}\in(\R^d)^{\mathcal{I}}$.

To connect the two drivers $G,g$, we also define a special type of composition of $g$, through its second argument, with a random vector. Let $\ell\in\cb{0,\ldots,K}$ and denote by $\A(t_{\ell})\subseteq 2^{\O}$ the partition of $\O$ that generates $\F^M_{t_{\ell}}$, which is of size $2^{\ell m}$. Let $\psi=(\psi_I)_{I\in\mathcal{I}}\in (\lcal_{t_{\ell}}^d)^{\mathcal{I}}$. Given $I\in\mathcal{I}$ and $A\in\A(t_{\ell})$, note that $\psi_I$ is constant on $A$ so that we can define $\psi^A_I\coloneqq \psi_I(\o)$ for some $\o\in A$, which is free of the choice of $\o\in A$, and write $\psi^A=(\psi^A_{I})_{I\in\mathcal{I}}\in(\R^d)^{\mathcal{I}}$.

For $k\in\cb{1,\ldots,K}$ and $\psi\in(\lcal^d_{t_{k-1}})^{\mathcal{I}}$, we define
	\[
	g\circ (t_{k-1},\psi) \coloneqq \sum_{A\in\A(t_{k-1})} g(t_{k-1},\psi^A)\ind_{A}.
	\]
	The next lemma states that this composition coincides with the nonlocal driver $G$ in \eqref{nonlocaldefn}.

\begin{lemma}\label{localvsnonlocal}
	For every $k\in\cb{1,\ldots,K}$ and $\psi=(\psi_I)_{I\in\mathcal{I}}\in (\lcal_{t_{k-1}}^d)^{\mathcal{I}}$, it holds
	\[
	G(t_{k-1},\psi) = g\circ (t_{k-1},\psi).
	\]
	\end{lemma}
	
\begin{proof}
	By the decomposability of $R$, we have
	\begin{align*}
	G(t_{k-1},\psi)&=\frac{1}{\Delta t_{k}}R_{t_{k-1},t_k}\of{-\sum_{I\in\mathcal{I}} \psi_I\Delta M_I(t_k)}\\
	&=\frac{1}{\Delta t_{k}}R_{t_{k-1},t_k}\of{\sum_{A\in\A(t_{k-1})} \of{- \sum_{I\in\mathcal{I}} \psi_I\Delta M_I(t_k)}\ind_{A}}\\
	&=\frac{1}{\Delta t_{k}}R_{t_{k-1},t_k}\of{\sum_{A\in\A(t_{k-1})}\of{ -\sum_{I\in\mathcal{I}} \psi^A_I\Delta M_I(t_k)}\ind_{A}}\\
	&=\frac{1}{\Delta t_{k}}\sum_{A\in\A(t_{k-1})}R_{t_{k-1},t_k}\of{ -\sum_{I\in\mathcal{I}} \psi^A_I\Delta M_I(t_k)}\ind_{A}\\
	&=\sum_{A\in\A(t_{k-1})} g(t_{k-1},\psi^A)\ind_{A}=g\circ (t_{k-1},\psi),
	\end{align*}
	from which the result follows.
	\end{proof}	
	
Thanks to Lemma~\ref{localvsnonlocal}, the one-step BS$\Delta$I in Corollary~\ref{nonlocalcor} can be rewritten as
\begin{align*}
&Y(t_{k-1})\in Y(t_k)+ g\circ(t_{k-1},\psi(t_{k-1}))\Delta t_k-\sum_{I\in\mathcal{I}}\psi_I(t_{k-1})\Delta M_I(t_k),\quad k\in\cb{1,\ldots,K},\\
& Y(t_K)\in -X+R_{t_K}(0).
\end{align*}
Note that this BS$\Delta$I is a \emph{functional} inclusion where one random vector is included in a set of random vectors. In Proposition~\ref{localprop} below, we present an alternative BS$\Delta$I that is a \emph{random} inclusion with a completely local driver $\hat{g}\colon\Omega\times\cb{t_0,\ldots,t_{K-1}}\times (\R^d)^{\mathcal{I}} \to \P_+(\R^d)$ defined by
\[
\hat{g}(\omega, t_{k-1},z) \coloneqq g(t_{k-1},z)(\omega)= \frac{1}{\Delta t_k} R_{t_{k-1},t_k}\of{-\sum_{I \in \mathcal{I}} z_I \Delta M_I(t_k)}(\omega)
\]
for each $\o\in\O,k\in\cb{1,\ldots,K},z=(z_I)_{I\in\mathcal{I}}\in(\R^d)^{\mathcal{I}}$. To begin with, we formulate the connection between the drivers $G$ and $\hat{g}$ in the next lemma.

\begin{lemma}\label{projectionrule}
	Let $k\in\cb{1,\ldots,K}$ and $\psi=(\psi_I)_{I\in\mathcal{I}}\in (\lcal_{t_{k-1}}^d)^{\mathcal{I}}$. For each $\o\in\O$, it holds
	\[
	G(t_{k-1},\psi)(\o)=\hat{g}(\o,t_{k-1},\psi(\o)).
	\]
	\end{lemma}

\begin{proof}
	Let $\o\in\O$. We claim that
	\[
	G(t_{k-1},\psi)(\o)=\of{\sum_{A\in\A(t_{k-1})} g(t_{k-1},\psi^A)\ind_{A}}(\o)=\hat{g}(\o,t_{k-1},\psi(\o)).
	\]
	The first equality is by the calculation in the proof of Lemma~\ref{localvsnonlocal}. To prove the second equality, let $\bar{A}$ be the unique set in $\A(t_{k-1})$ for which $\o\in \bar{A}$. Let $u\in (\sum_{A\in\A(t_{k-1})} g(t_{k-1},\psi^A)\ind_{A})(\o)$. Hence,
	\[
	u\ind_{\cb{\o}}\in \sum_{A\in\A(t_{k-1})} g(t_{k-1},\psi^A)\ind_{A}\ind_{\cb{\o}} = g(t_{k-1},\psi^{\bar{A}})\ind_{\bar{A}}\ind_{\cb{\o}}=g(t_{k-1},\psi^{\bar{A}})\ind_{\cb{\o}}.
	\]
	This shows that $u\in g(t_{k-1},\psi^{\bar{A}})(\omega)=g(t_{k-1},\psi(\o))(\o)=\hat{g}(\o,t_{k-1},\psi(\o))$. Conversely, let $u\in \hat{g}(\o,t_{k-1},\psi(\o))=g(t_{k-1},\psi^{\bar{A}})(\omega)$. So
	\[
	u\ind_{\cb{\o}}\in g(t_{k-1},\psi^{\bar{A}})\ind_{\cb{\omega}}=\of{\sum_{A\in\A(t_{k-1})} g(t_{k-1},\psi^A)\ind_{A}}\ind_{\cb{\omega}},
	\]
	which shows that $u\in (\sum_{A\in\A(t_{k-1})} g(t_{k-1},\psi^A)\ind_{A})(\o)$.
	\end{proof}

With Lemma \ref{projectionrule}, we are ready to formulate the random BS$\Delta$I in the next proposition.

\begin{proposition}\label{localprop}
	Let $X\in \lcal_{t_K}^d$ and consider a process $(Y(t))_{t\in\mathbb{T}}$ such that $Y(t_K)\in R_{t_K}(X)$ and $Y(t_{k-1})\in R_{t_{k-1},t_{k}}(-Y(t_{k}))$ for each $k\in\cb{1,\ldots,K}$. Then, there exists $\psi(t_{k})=(\psi_I(t_{k}))_{I\in\mathcal{I}}\in (\lcal_{t_{k}}^d)^{\mathcal{I}}$ for each $k\in\cb{1,\ldots,K}$ such that the random BS$\Delta$I
	\begin{align*}
	&Y(\o,t_{k-1})\in Y(\o,t_k)+ \hat{g}(\o,t_{k-1},\psi(\o,t_{k-1}))\Delta t_k-\sum_{I\in\mathcal{I}}\psi_I(\o,t_{k-1})\Delta M_I(\o,t_k),\; k\in\cb{1,\ldots,K},\\
	& Y(\o,t_K)\in -X(\o)+R_{t_K}(0)(\o)
	\end{align*}
	holds for every $\o\in\O$. Conversely, if there exist a $d$-dimensional adapted process $(Y(t))_{t\in\mathbb{T}}$ and $\psi(t_{k-1})\in (\lcal_{t_{k-1}}^d)^{\mathcal{I}}$ for each $k\in\cb{1,\ldots,K}$ such that the above random BS$\Delta$I holds for every $\o\in\O$, then $Y(t_K)\in R_{t_K}(X)$ and $Y(t_{k-1})\in R_{t_{k-1},t_k}(-Y(t_k))$ for each $k\in\cb{1,\ldots,K}$. In each case, the multi-step version of the random BS$\Delta$I
		\begin{align*}
		&Y(\o,t_k)\in Y(\o,t_K)+\sum_{\ell=k+1}^K \hat{g}(\o,t_{\ell-1},\psi(\o,t_{\ell-1}))\Delta t_\ell-\sum_{\ell=k+1}^K\sum_{I\in\mathcal{I}}\psi_I(\o,t_{\ell-1})\Delta M_I(\o,t_\ell)
		\end{align*}
	holds for every $k\in\cb{0,\ldots,K-1}$ and $\o\in\O$.
\end{proposition}
\begin{proof}
	Let $k\in\cb{0,\ldots,K-1}$. We rewrite the BS$\Delta$I in Corollary~\ref{nonlocalcor} as
	\[
	Y(t_{k-1})-Y(t_k)+\sum_{I\in\mathcal{I}}\psi_I(t_{k-1})\Delta M_I(t_k) \in  G(t_{k-1},\psi(t_{k-1}))\Delta t_{k}.
	\]
	By the decomposability of $R$, it follows that the set on the right is decomposable. Hence, the above inclusion is equivalent to 
	\[
	Y(\o,t_{k-1})-Y(\o,t_k)+\sum_{I\in\mathcal{I}}\psi_I(\o,t_{k-1})\Delta M_I(\o,t_k) \in  G(t_{k-1},\psi(t_{k-1}))(\o)\Delta t_{k},
	\]
	which is equivalent to the random one-step inclusion in the statement of the proposition, by the definition of $\hat{g}$. Similarly, it can be checked that the multi-step inclusion in Corollary~\ref{nonlocalcor} is equivalent to the random multi-step inclusion in the statement of the proposition. All the claims follow immediately from these equivalences and Corollary~\ref{nonlocalcor}.
	\end{proof}

\begin{remark}
    To consider the local analog of Corollary~\ref{multistepcor}, we utilize a notation from~\cite{LR14,algorithm} on the \emph{finitely generated} filtration considered herein. As this filtration is generated by the Bernoulli random walk, we can view our processes as taking value on a tree.  With that in mind, we present this local analog of Corollary~\ref{multistepcor} using the concept and notation of a tree. Recall that $\A(t_\ell) \subseteq 2^\O$ denotes the set of all atoms of $\F^M_{t_\ell}$ for $\ell \in \{0,\ldots,K\}$ and, for a given $Z\in \lcal_{t_\ell}^d$, we denote by $Z^A$ the constant value of $Z$ on $A\in\A(t_\ell)$. Additionally,  when $\ell\leq K-1$, we denote the ``successor'' nodes at time $t_{\ell+1}$ of $A_\ell\in\A(t_\ell)$ by
    \[\succ(A_{\ell}) = \cb{A_{\ell+1} \in \A(t_{\ell+1}) \; | \; A_{\ell+1} \subseteq A_\ell}.\]

    Fix $X \in \lcal_{t_K}^d$.  In this atomized notation, the one-step random BS$\Delta$I from Proposition~\ref{localprop} for $k \in \cb{1,\ldots,K}$ can be rewritten as
    \begin{align*}
    Y^{A_{k-1}}(t_{k-1}) &\in \bigcap_{A_{k} \in \succ(A_{k-1})} \left[Y^{A_{k}}(t_{k}) + \hat{g}^{A_{{k-1}}}(t_{k-1},\psi^{A_{{k-1}}}(t_{k-1}))\Delta t_{k-1} - \sum_{I \in \mathcal{I}} \psi^{A_{{k-1}}}_I(t_{k-1})\Delta M_I^{A_{k}}(t_{k})\right]\\
    Y^{A_{K}}(t_K) &\in -X^{A_{K}}+ R_{t_K}(0)(\o_K)
    \end{align*}
    for every $A_{k-1} \in \A(t_{k-1})$.  In this way, we can consider the multi-step inclusion at time $t_k$ with $k \in \cb{0,\ldots,K-1}$ and $A_k\in\A(t_k)$
    \begin{align}\label{multistepincR}
    y &\in \bigcap_{\substack{A_{k+1} \in \succ(A_k) \\ \vdots \\ A_K \in \succ(A_{K-1})}} \left[Y^{A_K}(t_K) + \sum_{\ell = k+1}^K \hat{g}^{A_{\ell-1}}(t_{\ell-1},\psi^{A_{\ell-1}}(t_{\ell-1})) \Delta t_{\ell} - \sum_{\ell = k+1}^K \sum_{I \in \mathcal{I}} \psi^{A_{\ell-1}}_I(t_{\ell-1}) \Delta M^{A_{\ell}}_I(t_\ell)\right]
    \end{align}
    for some adapted process $\psi(t_{\ell-1}) \in (\lcal_{t_{\ell-1}}^d)^{\mathcal{I}}$. (This representation is an example of the conditional core introduced in~\cite{conditionalcore}; we conjecture that the conditional core may be useful for extending \eqref{multistepincR} to, e.g., the continuous time limit.)	Similar to Corollary~\ref{multistepcor}, given $A_k\in\A(t_k)$, it can be checked that the set $R^{A_k}_{t_k}(X)=R_{t_k}(X)(\o)$, $\o\in A_k$, coincides with the reachable set of \eqref{multistepincR}, that is,
	\begin{align*}
	&R^{A_k}_{t_k}(X)= \left\{y\in\R^d \; \left| \; \begin{array}{l}\text{\eqref{multistepincR} holds for some } Y^{A_{K}}(t_K) \in -X^{A_K} + R^{A_K}_{t_K}(0), \\ \psi^{A_k}(t_{k}),\ldots,\psi^{A_{K-1}}(t_{K-1})\in(\R^d)^{\mathcal{I}}, \\ \forall A_{k+1}\in\succ(A_k),\ldots,A_{K}\in\succ(A_{K-1}) \end{array}\right.\right\}.
	\end{align*}
\end{remark}

We conclude this section with two examples on superhedging, as a follow-up of the model introduced in Example \ref{exmodel1}.

\begin{example}\label{exmodel2}
	In the setting of Example \ref{exmodel1}, let us calculate the semi-local driver $g^{SHP}$ of the superhedging risk measure $R^{SHP}$ when $d=2$ and $m=1$. Following \eqref{semilocal} and the definitions in Example \ref{exmodel1}, for each $t\in\cb{0,\ldots,T-1}$ and $z\in\R^2$, we have
	\begin{align*}
	&g^{SHP}(t,z)\\
	&=R^{SHP}_{t,t+1}(-z\Delta M(t+1))=SHP_t(z B(t+1))\\
	&=\cb{Y\in\lcal_t^2 \; | \; Y\in z B(t+1)+\sum_{s=t}^T L_d^0(\F^M_s,K_s)}\\
	&=L_d^0(\F^M_t,K_t)+\cb{Y\in\lcal_t^2 \; | \; Y\in z B(t+1)+\sum_{s=t+1}^T L_d^0(\F^M_s,K_s)}\\
	&=L_d^0(\F^M_t,K_t)+\Bigg\{Y\in\lcal_t^2 \; | \; \begin{pmatrix} Y^1\\Y^2\end{pmatrix}=\begin{pmatrix} z^1\\z^2\end{pmatrix}B(t+1)+\sum_{s=t+1}^T \sqb{b_s\begin{pmatrix} -S_s^b\\1\end{pmatrix}+ a_s\begin{pmatrix} S_s^a\\-1\end{pmatrix}},\\
	&\quad\quad\quad\quad\quad\quad\quad\quad\quad \quad\quad\quad\quad    a_s,b_s\in \lcal_{s,+}^2, s\in\cb{t+1,\ldots,T}\Bigg\}.
	\end{align*}
	Note that for each $s\in\cb{t+1,\ldots,T}$, we have
	\[
	S_s^b=S_t^b e^{\sigma B(t+1)}e^{\sigma(M(s)-M(t+1))},\quad S_s^a=S_t^a e^{\sigma B(t+1)}e^{\sigma(M(s)-M(t+1))}.
	\]
	Let us denote by $\hat{g}^{SHP}$ the local driver. By the structure of the solvency cones, it follows that $\hat{g}^{SHP}(\o,t,z)$ depends on $\o$ only through $M(t)(\o)$. Hence, similar to \eqref{Ktnotation}, we write
	\[
	\hat{g}^{SHP}(\o,t,z)=\hat{g}^{SHP}(c,t,z)
	\]
	whenever $M(t)(\o)=c$. The above calculations for the semi-local driver $g^{SHP}$ imply that
	\begin{align*}
	\hat{g}^{SHP}(c,t,z)&=K_t(c)+\Bigg[\Bigg(z+\sum_{s=t+1}^T \bigcap_{b(t+2),\ldots,b(s)\in\cb{-1,+1}}K_s(c+1+b(t+2)+\ldots+b(s))\Bigg)\\
	&\quad\quad\quad\quad \cap \Bigg(-z+\sum_{s=t+1}^T \bigcap_{b(t+2),\ldots,b(s)\in\cb{-1,+1}}K_s(c-1+b(t+2)+\ldots+b(s))\Bigg)\Bigg].
	\end{align*}
	Note that, for each $s\in\cb{t+1,\ldots,T}$, the intersections over $b(t+2),\ldots,b(s)\in\cb{-1,+1}$ has $s-t$ different cones due to the dependence on the sum $b(t+2)+\ldots+b(s)$.
	
	As a special case, let us assume that $T=2$. Then, the general expression for $\hat{g}^{SHP}$ yields
	\begin{align*}
	&\hat{g}^{SHP}(0,0,z)=K_0+[(z+K_1(1)+(K_2(2)\cap K_2(0)))\cap(-z+K_1(-1)+(K_2(0)\cap K_2(-2)))],\\
	&\hat{g}^{SHP}(1,1,z)=K_1(1)+[(z+K_2(2))\cap(-z+K_2(0))],\\
	&\hat{g}^{SHP}(-1,1,z)=K_1(-1)+[(z+K_2(0))\cap(-z+K_2(-2))].
	\end{align*}
	
	In this example, the driver is the intersection of two ``affine" set-valued functions of $z$. In the univariate setting $d=m=1$ with no transaction costs, it is worth noting that the hedging example considered in \citet[Section 3]{G06} has a driver that is linear in $z$.

\end{example}

\begin{example}\label{exmodel3}
	The calculations in Example \ref{exmodel2} can be extended to the case with $d=3$ and $m=2$. In this case, in addition to the principle random walks $M_1,M_2$, we also need the process $M_{\cb{1,2}}$ whose increments are of the form $\Delta M_{\cb{1,2}}(t)=B_{\cb{1,2}}(t)=B_1(t)B_2(t)$. Let $t\in\cb{0,1,\ldots,T-1}$ and $z=(z_1,z_2,z_{12})\in(\R^3)^{\mathcal{I}}$. For the semi-local driver, similar calculations as in Example \ref{exmodel2} yield
	\[
	g^{SHP}(t,z)=\cb{Y\in\lcal_t^3 \; | \; Y\in z_1B_1(t+1)+z_2B_2(t+1)+z_{12}B_{\cb{1,2}}(t+1)+\sum_{s=t}^T L_d^0(\F_s^M,K_s)}.
	\]
	
	For each $u\in\N$, let us define the set 
	\[
	\mathbb{S}(u)\coloneqq\cb{\sum_{s=1}^u b(s) \; | \; b(s)\in\cb{-1,+1}, s\in\cb{1,\ldots,u}}.
	\]
	For instance, $\mathbb{S}(1)=\cb{-1,+1}$, $\mathbb{S}(2)=\cb{-2,0,+2}$, $\mathbb{S}(3)=\cb{-3,-1,+1,+3}$. For convenience, let us also set $\mathbb{S}(0)=\cb{0}$. Let $\o\in\O$. To express the local driver $\hat{g}^{SHP}$, we write
	\[
	K_t(\o)=K_t(c_1,c_2),\quad\hat{g}^{SHP}(\o,t,z)=\hat{g}^{SHP}(c_1,c_2,t,z)
	\]
	whenever $M_1(t)(\o)=c_1$ and $M_2(t)(\o)=c_2$. Then, $\hat{g}^{SHP}$ is given by
	\begin{align*}
	&\hat{g}^{SHP}(c_1,c_2,t,z)=K_t(c_1,c_2)\\
	&\quad\quad\quad\quad\quad\quad\quad\quad\quad +\Bigg[\Bigg( z_1+z_2+z_{12}+\sum_{s=t+1}^T \bigcap_{b_1,b_2\in \mathbb{S}(s-t-1)}K_s\of{c_1+1+b_1 ,c_2+1+b_2}\Bigg)	\\
	&\quad\quad\quad\quad\quad\quad\quad\quad\quad \cap\Bigg(z_1-z_2-z_{12}+\sum_{s=t+1}^T \bigcap_{b_1,b_2\in \mathbb{S}(s-t-1)}K_s\of{c_1+1+b_1 ,c_2-1+b_2}\Bigg)	\\
	&\quad\quad\quad\quad\quad\quad\quad\quad\quad \cap \Bigg(-z_1+z_2-z_{12}+\sum_{s=t+1}^T \bigcap_{b_1,b_2\in \mathbb{S}(s-t-1)}K_s\of{c_1-1+b_1 ,c_2+1+b_2}	\Bigg)\\
	&\quad\quad\quad\quad\quad\quad\quad\quad\quad \cap\Bigg(-z_1-z_2+z_{12}+\sum_{s=t+1}^T \bigcap_{b_1,b_2\in \mathbb{S}(s-t-1)}K_s\of{c_1-1+b_1 ,c_2-1+b_2}\Bigg)\Bigg].
	\end{align*}
	As a special case, for $T=2$, we obtain
	\begin{align*}
	&\hat{g}^{SHP}(0,0,0,z)\\
	&=K_0+\Big[\big(z_1+z_2+z_{12}+K_1(1,1)+(K_2(2,2)\cap K_2(2,0)\cap K_2(0,2)\cap K_2(0,0))\big)\\
	&\quad\quad\quad \cap\big(z_1-z_2-z_{12}+K_1(1,-1)+(K_2(2,0)\cap K_2(2,-2)\cap K_2(0,0)\cap K_2(0,-2))\big)\\
	&\quad\quad\quad \cap\big(-z_1+z_2-z_{12}+K_1(-1,1)+(K_2(0,2)\cap K_2(-2,2)\cap K_2(0,0)\cap K_2(-2,0))\big)\\
	&\quad\quad\quad \cap\big(-z_1-z_2+z_{12}+K_1(-1,-1)+(K_2(0,0)\cap K_2(0,-2)\cap K_2(-2,0)\cap K_2(-2,-2))\big)\Big]
	\end{align*}
	for $t=0$ and
	\begin{align*}
	&\hat{g}^{SHP}(1,1,1,z)=K_1(1,1)+\Big[(z_1+z_2+z_{12}+K_2(2,2))\cap(z_1-z_2-z_{12}+K_2(2,0))\\
	&\quad\quad\quad\quad\quad\quad\quad\quad\quad\quad\quad\quad\quad \cap(-z_1+z_2-z_{12}+K_2(0,2))\cap(-z_1-z_2+z_{12}+K_2(0,0))\Big],\\
	&\hat{g}^{SHP}(1,-1,1,z)=K_1(1,-1)+\Big[(z_1+z_2+z_{12}+K_2(2,0))\cap(z_1-z_2-z_{12}+K_2(2,-2))\\
	&\quad\quad\quad\quad\quad\quad\quad\quad\quad\quad\quad\quad\quad \cap(-z_1+z_2-z_{12}+K_2(0,0))\cap(-z_1-z_2+z_{12}+K_2(0,-2))\Big],\\
	&\hat{g}^{SHP}(-1,1,1,z)=K_1(-1,1)+\Big[(z_1+z_2+z_{12}+K_2(0,2))\cap(z_1-z_2-z_{12}+K_2(0,0))\\
	&\quad\quad\quad\quad\quad\quad\quad\quad\quad\quad\quad\quad\quad \cap(-z_1+z_2-z_{12}+K_2(-2,2))\cap(-z_1-z_2+z_{12}+K_2(-2,0))\Big],\\
	&\hat{g}^{SHP}(-1,-1,1,z)=K_1(-1,-1)+\Big[(z_1+z_2+z_{12}+K_2(0,0))\cap(z_1-z_2-z_{12}+K_2(0,-2))\\
	&\quad\quad\quad\quad\quad\quad\quad\quad\quad\quad\quad\quad\quad \cap(-z_1+z_2-z_{12}+K_2(-2,0))\cap(-z_1-z_2+z_{12}+K_2(-2,-2))\Big]
	\end{align*}
	for $t=1$.
	\end{example}

\section{Set-Valued Backward Stochastic Difference Equation}\label{sec:sv-bsde}

In this section, we show that a given dynamic set-valued risk measure in discrete time gives rise to a set-valued backward stochastic difference equation (SV-BS$\Delta$E). In contrast to the BS$\Delta$I presented in the prior section, which is evaluated for the selectors of the solution set, a SV-BS$\Delta$E provides the dynamics for the entire set as a ``singular'' object.  As depicted in~\eqref{eq:mptc-recursive}, multiportfolio time consistency provides a recursive relation for the \emph{set} of acceptable capital allocations.  Therefore, since the dynamic programming principle for set-valued risk measures is defined for the full set rather than selectors, we wish to consider the SV-BS$\Delta$E which encodes the defining relation of multiportfolio time consistency for dynamic risk measures.

We wish to highlight that the general theory of set-valued difference and differential equations are typically restricted to the space of compact and convex sets.  However, as previously discussed, risk measures are naturally upper sets and thus require further consideration.  As far as the authors are aware, the method for constructing a set-valued difference (or differential in the limit) equation via an intersection of halfspaces is novel to this work.

In the setting of Section~\ref{sec:discrete}, we work with the filtered probability space $(\O,\F, (\F^M_{t})_{t\in\mathbb{T}}, \Pr)$, where $M$ is the $m$-dimensional Bernoulli random walk defined in \eqref{randomwalkdefn}. Let us consider a multiportfolio time consistent dynamic set-valued convex risk measure $R=(R_{t})_{t\in\mathbb{T}}$ with one-step conditional risk measures $R_{t_{k-1},t_{k}}\colon \lcal_{t_{k}}^d\to \P_+(\lcal_{t_{k-1}}^d)$, $k\in\cb{1,\ldots,K}$. For the terminal risk measure $R_{t_K}\colon \lcal_{t_K}^d\to \P_+(\lcal_{t_K}^d)$, note that we have $R_{t_K}(X)=-X+R_{t_K}(0)$ for each $X\in \lcal_{t_K}^d$.

As with our discussion of the BS$\Delta$I, we relate $R$ to a SV-BS$\Delta$E with a \emph{nonlocal} driver. To that end, let us introduce the domain
\[
\mathbb{D}_E \coloneqq \cb{(t_k,\psi,w) \; | \; k \in \cb{0,\ldots,K-1}, \; \psi=(\psi_I)_{I\in\mathcal{I}} \in (\lcal_{t_k}^d)^{\mathcal{I}}, \; w \in \lcal_{t_{k+1},+}^d}.
\]
Note that this domain has an additional variable than that for the driver of the BS$\Delta$I introduced by \eqref{dominc} in the prior section.
With this, define the set-valued driver $G_E: \mathbb{D}_E\to 2^{\lcal^d}$ by
\begin{equation}\label{eq:SV-driver}
G_E(t_{k-1},\psi,w) := \frac{1}{\Delta t_k} R_{t_{k-1},t_k}\left[-\sum_{I \in \mathcal{I}} \psi_I \Delta M_I(t_k) - \Gamma_{t_k}(w)\right]
\end{equation}
for each $(t_{k-1},\psi,w) \in \mathbb{D}_E$ with $k \in \{1,...,K\}$ and where 
\begin{equation}\label{condhs}
\Gamma_{t_k}(w) := \{u \in \lcal_{t_k}^d \; | \; w^\T u \geq 0\}.
\end{equation}
Note that $G_E$ is adapted such that $G_E(t,\psi,w) \in \P_+(\lcal_t^d)$ whenever $(t,\psi,w) \in \mathbb{D}_E$.

The following lemma is required for the results about SV-BS$\Delta$Es. We separate it from the main results for readability.
\begin{lemma}\label{prop:H}
	Let $X \in \lcal_{t_K}^d$ and define
	\[
	H_{t_k}(X,w_{t_k}) := \cl[R_{t_{k}}(X) + \Gamma_{t_{k}}(w_{t_{k}})]
	\]
	for $w_{t_{k}} \in \lcal_{t_{k},+}^d$, where the conditional halfspace $\Gamma_{t_{k}}(w_{t_{k}})$ is defined by \eqref{condhs}. Then, for any $k \in \{1,\ldots,K\}$, the relation
	\begin{equation}\label{RofH}
	R_{t_{k-1}}(X) = \bigcap_{w_{t_k} \in \lcal_{t_k,+}^d} R_{t_{k-1},t_k}\left[-H_{t_k}(X,w_{t_k})\right]
	\end{equation}
	holds.
\end{lemma}
\begin{proof}
	First, notice that $H_{t_k}(X,w_{t_k}) = \{u \in \lcal_{t_k}^d \; | \; w_{t_k}^\T u \geq \rho_{t_k}(X,w_{t_k})\}$ for the scalarization
	\[
	\rho_{t_k}(X,w_{t_k}) \coloneqq \underset{u \in R_{t_k}(X)}{\essinf}w_{t_k}^\T u.
	\]
	Therefore, by \citet[Lemma~3.18]{comparison},
	\[
	R_{t_k}(X) = \bigcap_{w_{t_k} \in \lcal_{t_k,+}^d} H_{t_k}(X,w_{t_k}).
	\]
	To prove the equality in \eqref{RofH}, we proceed as follows.
	\begin{itemize}
		\item[$\subseteq$] By multiportfolio time consistency, monotonicity, and the above notes,
		\begin{align*}
		R_{t_{k-1}}(X) &= R_{t_{k-1},t_k}\left[-R_{t_k}(X)\right]\\
		&= R_{t_{k-1},t_k}\left[\bigcap_{w_{t_k} \in \lcal_{t_k,+}^d} -H_{t_k}(X,w_{t_k})\right]\\ 
		&\subseteq \bigcap_{w_{t_k} \in \lcal_{t_k,+}^d} R_{t_{k-1},t_k}\left[-H_{t_k}(X,w_{t_k})\right].
		\end{align*}
		\item[$\supseteq$] To get a contradiction, suppose that there exists $m \not\in R_{t_{k-1}}(X)$ with
		\[
		m \in \bigcap_{w_{t_k} \in \lcal_{t_k,+}^d} R_{t_{k-1},t_k}[-H_{t_k}(X,w_{t_k})].
		\]
		In particular, for every $w_{t_k} \in \lcal_{t_k,+}^d$, there exists $Z(w_{t_k}) \in H_{t_k}(X,w_{t_k})$ such that
		\begin{equation}\label{contr}
		m \in R_{t_{k-1},t_k}(-Z(w_{t_k})).
		\end{equation} 
		Before continuing, first we need to introduce some notation.  Let $\mathcal{M}$ be the space of all probability measures that are absolutely continuous with respect to $\Pr$.  For any $\Q \in \mathcal{M}^d$, $v \in \lcal_{t_{k-1}}^d$ and $\ell\in\cb{k,\ldots,K}$, define $w_{t_{k-1}}^{t_\ell}(\Q,v) \in \lcal_{t_\ell}^d$ by
		\[w_{t_{k-1}}^{t_\ell}(\Q,v)_i = \begin{cases} \frac{v_i\E\left[\frac{d\Q_i}{d\Pr} \; | \; \F^M_{t_\ell}\right]}{\E\left[\frac{d\Q_i}{d\Pr} \; | \; \F^M_{t_{k-1}}\right]} &\text{on } \cb{\E\left[\frac{d\Q_i}{d\Pr} \; | \; \F^M_{t_{k-1}}\right] > 0} \\ v_i &\text{on } \cb{\E\left[\frac{d\Q_i}{d\Pr} \; | \; \F^M_{t_{k-1}}\right] = 0}. \end{cases}\] 
		Define the space of set-valued dual variables $\W_{t_{k-1}} := \{(\Q,v) \in \mathcal{M}^d \times \lcal_{t_{k-1},+}^d \; | \; w_{t_{k-1}}^{t_K}(\Q,v) \in \lcal_{t_K,+}^d\}$ and the projection $\W_{t_{k-1}}(v) := \{\Q \in \mathcal{M}^d \; | \; (\Q,v) \in \W_{t_{k-1}}\}$.
		Note that for any $(\Q,v) \in \W_{t_{k-1}}$, we have $w_{t_{k-1}}^{t_k}(\Q,v)\in \lcal_{t_k,+}^d$ so that $Z(w_{t_{k-1}}^{t_k}(\Q,v))\in H_{t_k}(X,w_{t_{k-1}}^{t_k}(\Q,v))$, that is,
		\begin{equation}\label{comp}
		w_{t_{k-1}}^{t_k}(\Q,v)^\T Z(w_{t_{k-1}}^{t_k}(\Q,v)) \geq \rho_{t_k}(X,w_{t_{k-1}}^{t_k}(\Q,v))
		\end{equation}
		for every $(\Q,v)\in \W_{t_{k-1}}$. From \citet[Lemma~3.18]{comparison} and \citet[Theorem~5.3]{tcscalarmv},
		\begin{align*}
		& m\notin R_{t_{k-1}}(X)=\bigcap_{v\in\lcal_{t_{k-1},+}^d}H_{t_{k-1}}(X,v)\\
		&\Rightarrow\;\exists v \in \lcal_{t_{k-1},+}^d: \; \Pr(v^\T m < \rho_{t_{k-1}}(X,v)) > 0 \\
		&\Rightarrow\; \exists v \in \lcal_{t_{k-1},+}^d:\\ &\qquad \Pr\left(v^\T m < \underset{\Q \in \W_{t_{k-1}}(v)}{\esssup} \left(-\alpha_{t_{k-1},t_k}(\Q,v) + \E\sqb{\rho_{t_k}(X,w_{t_{k-1}}^{t_k}(\Q,v)) \; | \; \F^M_{t_{k-1}}}\right)\right) > 0.
        \end{align*}
        This implies that there exists some $(\Q^*,v) \in \W_{t_{k-1}}$ such that
        \[
		\Pr\left(v^\T m < -\alpha_{t_{k-1},t_k}(\Q^*,v) + \E\sqb{\rho_{t_k}(X,w_{t_{k-1}}^{t_k}(\Q^*,v)) \; | \; \F^M_{t_{k-1}}}\right) > 0;
        \]
        fix this choice of dual variables.  It then follows that
        \begin{align*}
		&\Pr\left(v^\T m < -\alpha_{t_{k-1},t_k}(\Q^*,v) + \E\sqb{\rho_{t_k}(X,w_{t_{k-1}}^{t_k}(\Q^*,v)) \; | \; \F^M_{t_{k-1}}}\right) > 0\\
		&\Rightarrow\; \Pr\left(v^\T m < -\alpha_{t_{k-1},t_k}(\Q^*,v) + \E\sqb{w_{t_{k-1}}^{t_k}(\Q^*,v)^\T Z(w_{t_{k-1}}^{t_k}(\Q^\ast,v))  \; | \; \F^M_{t_{k-1}}}\right) > 0\quad \text{(by \eqref{comp})} \\
		&\Rightarrow\; \Pr\left(v^\T m < -\alpha_{t_{k-1},t_k}(\Q^*,v) + v^\T \E^{\Q^*}\sqb{Z(w_{t_{k-1}}^{t_k}(\Q^\ast,v))  \; | \; \F^M_{t_{k-1}}}\right) > 0 \\ 
		&\Rightarrow\; \Pr\left(v^\T m < \underset{\Q \in \W_{t_{k-1}}(v)}{\esssup} \left(-\alpha_{t_{k-1},t_k}(\Q,v) + v^\T \E^{\Q}\sqb{Z(w_{t_{k-1}}^{t_k}(\Q^\ast,v)) \; | \; \F^M_{t_{k-1}}}\right)\right) > 0 \\
		&\Rightarrow\; \Pr\left(v^\T m < \rho_{t_{k-1}}^v(-Z(w_{t_{k-1}}^{t_k}(\Q^\ast,v)))\right) > 0\\
		&\Rightarrow\; m \notin R_{t_{k-1}}(-Z(w_{t_{k-1}}^{t_k}(\Q^*,v))),
		\end{align*}
		which is a contradiction to \eqref{contr}.
	\end{itemize}
\end{proof}

The next proposition provides a SV-BS$\Delta$E associated to $R$, which is analogous to the BS$\Delta$I in Proposition~\ref{localprop}. The ``difference" in the SV-BS$\Delta$E is formulated by the notion of geometric difference between sets: for $C,D\subseteq \lcal^d$, we define their geometric difference by
\[
C -^{\scalebox{0.5}{$\bullet$}} D := \{u \in \lcal^d \; | \; u + D \subseteq C\}.
\]
Fix $k\in\cb{1,\ldots,K}$ and suppose that $C,D\subseteq\lcal_{t_k}^d$. Often we want to consider the difference modified by the halfspace
\[
(\cl[C+\Gamma_{t_k}(w)] -^{\scalebox{0.5}{$\bullet$}} D) \cap \lcal_{t_{k-1}}^d
\]
for some $w \in \lcal_{t_k,+}^d$.
It can be checked that
\begin{align*}
(\cl[C+\Gamma_{t_k}(w)]-^{\scalebox{0.5}{$\bullet$}} D) \cap\lcal_{t_{k-1}}^d &= (\cl[C+\Gamma_{t_k}(w)]-^{\scalebox{0.5}{$\bullet$}}\cl[D+\L^d_{t_{k},+}])\cap\lcal_{t_{k-1}}^d\\
&=(\cl[C+\Gamma_{t_k}(w)]-^{\scalebox{0.5}{$\bullet$}}\cl[D+\Gamma_{t_k}(w)])\cap\lcal_{t_{k-1}}^d.
\end{align*}

\begin{proposition}\label{lemma:SV-BSDE}
	Let $X \in \lcal_{t_K}^d$. 
	If $(\ycal(t))_{t \in \mathbb{T}}$ is a process such that $\ycal(t_k) = R_{t_k}(X)$ for each $k \in \{0,\ldots,K\}$, then the SV-BS$\Delta$E 
	\begin{align*}
	\ycal(t_{k-1}) &= \bigcap_{\substack{w_{t_k} \in \lcal_{t_k,+}^d,\\\psi \in \Psi(t_{k-1},w_{t_k})}} \left[\begin{array}{l}G_E(t_{k-1},\psi,w_{t_k}) \Delta t_k\\ + \left(\cl\left[\ycal(t_k) + \Gamma_{t_k}(w_{t_k})\right] -^{\scalebox{0.5}{$\bullet$}} \left[\sum_{I \in \mathcal{I}} \psi_I \Delta M_I(t_k) + \lcal_{t_k,+}^d \right]\right) \cap \lcal_{t_{k-1}}^d\end{array}\right],\\
	\ycal(t_K) &= -X + R_{t_K}(0)
	\end{align*}
	holds, where
	\begin{equation}\label{const}
	\Psi(t_{k-1},w_{t_k}) \coloneqq \bigcup_{\xi \in \lcal_{t_{k-1}}^d} \left\{\psi \in (\lcal_{t_{k-1}}^d)^{\mathcal{I}} \; | \; \xi + \sum_{I \in \mathcal{I}} \psi_I \Delta M_I(t_k) + \Gamma_{t_k}(w_{t_k}) = \cl[\ycal(t_k) + \Gamma_{t_k}(w_{t_k})]\right\}.
	\end{equation}
	(In the above SV-BS$\Delta$E, we make the convention that the intersection over $\Psi(t_{k-1},w_{t_k}) $ gives $\lcal^d_{t_{k-1}}$ in case $\Psi(t_{k-1},w_{t_k}) =\emptyset$.)
	
	Conversely, if there exists a set-valued process $(\ycal(t))_{t \in \mathbb{T}}$ and a set $\Psi(t_{k-1},w_{t_k}) \subseteq (\lcal_{t_{k-1}}^d)^{\mathcal{I}}$ for each $k \in \{1,\ldots,K\}$ and $w_{t_k} \in \lcal_{t_k,+}^d$ such that the above SV-BS$\Delta$E holds, then $\ycal(t_k) = R_{t_k}(X)$ for each $k \in \{0,\ldots,K\}$.  
\end{proposition}
\begin{proof}
	Let $(\ycal(t))_{t \in \mathbb{T}}$ be a process such that $\ycal(t_k) = R_{t_k}(X)$ for each $k \in \{0,\ldots,K\}$. By construction of the risk measure, the terminal condition of the SV-BS$\Delta$E trivially holds.
	Consider now $k \in \{1,\ldots,K\}$. If $\ycal(t_k) = \emptyset$ then, by multiportfolio time consistency, $\ycal(t_{k-1}) = \emptyset$ as well. As such SV-BS$\Delta$E is satisfied trivially prior to $t_k$; for the remainder of this proof, we will assume that $\ycal(t_k) \neq \emptyset$.  
	
	Utilizing Lemma~\ref{prop:H}, where $H_{t_k}(X,w_{t_k}) = \cl[\ycal(t_k) + \Gamma_{t_k}(w_{t_k})]$, we get
	\[
	\ycal(t_{k-1}) = \bigcap_{w_{t_k} \in \lcal_{t_k,+}^d} R_{t_{k-1},t_k}[-H_{t_k}(X,w_{t_k})].
	\]
	Let
	\begin{equation}\label{Wdefn}
	W_{t_k}(X)\coloneqq \cb{w\in\lcal_{t_k,+}^d \; | \; \forall\omega\in\O\colon H_{t_k}(X,w_{t_k})(\omega)\neq \R^d}.
	\end{equation}
	We claim that
	\begin{equation}\label{W(X)}
	\ycal(t_{k-1}) = \bigcap_{w_{t_k} \in W_{t_k}(X)} R_{t_{k-1},t_k}[-H_{t_k}(X,w_{t_k})].
	\end{equation}
	Indeed, the $\subseteq$ part of the above equality is clear. To prove the $\supseteq$ part, let $w_{t_k}\in \lcal_{t_k,+}^d\setminus W_{t_k}(X)$. Recall from the previous section that $\A(t_k)$ denotes the partition of $\O$ that generates $\F_{t_k}^M$. Given $A\in\A(t_k)$, note that $\ycal(t_{k})(\omega)$ is the same nonempty subset of $\R^d$ for all $\o\in A$; let us denote this set by $\ycal^A(t_k)$.  First, we will assume that $\ycal^A(t_k) \neq \R^d$. Since $\ycal^A(t_k)\notin\cb{\emptyset,\R^d}$, there exists $w^A\in\R^d_+\sm\cb{0}$ such that
	\[
	\inf_{m\in \ycal^A(t_k)}(w^A)^\T m \in\R.
	\]
	Similarly, let us denote by $H_{t_k}^A(X,w_{t_k})=H_{t_k}(X,w_{t_k})(\o)$ for all $\o\in A$, $A\in\A(t_k)$. Since $w_{t_k}\notin W_{t_k}(X)$, the set $\mathcal{B}=\cb{A\in\A(t_k) \; | \; H_{t_k}^A(X,w_{t_k})=\R^d}$ is nonempty. Let
	\[
	\tilde{w}_{t_k}\coloneqq \sum_{A\in \mathcal{B}}w^A\ind_{A}+\sum_{A\in\A(t_k)\backslash\mathcal{B}}w_{t_k}\ind_A \in W_{t_k}(X)
	\]
	Then, it can be checked that $H_{t_k}(X,w_{t_k})\supseteq H_{t_k}(X,\tilde{w}_{t_k})$. By monotonicity, it follows that
	\[
	R_{t_{k-1},t_k}[-H_{t_k}(X,w_{t_k})]\supseteq R_{t_{k-1},t_k}[-H_{t_k}(X,\tilde{w}_{t_k})].
	\]
	Hence, the $\supseteq$ part of \eqref{W(X)} follows so long as $\ycal^A(t_k) \neq \R^d$ for every $A \in \A(t_k)$.  This must be true due to the construction of $\ycal(t_k)$.  In particular,
    \[\ycal(t_k) := R_{t_k}(X) \subseteq R_{t_k}(\|X\|_{\infty}) = R_{t_k}(0) - \|X\|_{\infty}\]
    due to monotonicity and translativity of the risk measure with 
    \[\|X\|_{\infty} = \left(\max_{\omega \in \Omega} |X_1(\omega)| \; , \; \ldots \; , \; \max_{\omega \in \Omega} |X_d(\omega)|\right).\]
    Therefore, for any $A \in \A(t_k)$ and $\omega \in A$, $\ycal^A(t_k) = \R^d$ implies $R_{t_k}(0)(\omega) = \R^d$.  However, this violates normalization of the risk measure and, as such, $\ycal^A(t_k) \neq \R^d$ for every $A \in \A(t_k)$.

Let $w_{t_k}\in W_{t_k}(X)$. By construction of $H_{t_k}(X,w_{t_k})$ and Lemma~\ref{genpredrep}, $H_{t_k}(X,w_{t_{k}})$ has the predictable representation
\begin{equation}
H_{t_k}(X,w_{t_{k}}) 
= \left[\xi(t_{k-1},w_{t_k}) + \sum_{I \in \mathcal{I}} \psi_I(t_{k-1},w_{t_k})\Delta M_I(t_k)\right] + \Gamma_{t_k}(w_{t_k})
\end{equation}
for some $\xi(t_{k-1},w_{t_k}),\psi_I(t_{k-1},w_{t_k}) \in \lcal_{t_{k-1}}^d, I \in \mathcal{I}$. Consider the set $\Psi(t_{k-1},w_{t_k})$ of all constructors of $H_{t_k}(X,w_{t_k})$ defined by \eqref{const}. By construction, for $\psi \in \Psi(t_{k-1},w_{t_k})$, there exists some $\xi(t_{k-1},w_{t_k},\psi) \in \lcal^d_{t_{k-1}}$ such that $H_{t_k}(X,w_{t_k}) = \xi(t_{k-1},w_{t_k},\psi) + \sum_{I \in \mathcal{I}} \psi_I \Delta M_I(t_k) + \Gamma_{t_k}(w_{t_k})$. Therefore,
\begin{align*}
&R_{t_{k-1},t_k}[-H_{t_k}(X,w_{t_k})]\\
 &= \bigcap_{\psi \in \Psi(t_{k-1},w_{t_k})} R_{t_{k-1},t_k}\left[-\xi(t_{k-1},w_{t_k},\psi) - \sum_{I \in \mathcal{I}} \psi_I \Delta M_I(t_k) - \Gamma_{t_k}(w_{t_k})\right] \\
&= \bigcap_{\psi \in \Psi(t_{k-1},w_{t_k})} R_{t_{k-1},t_k}\left[-\xi(t_{k-1},w_{t_k},\psi) - \sum_{I \in \mathcal{I}} \psi_I \Delta M_I(t_k) - \Gamma_{t_k}(w_{t_k}) - \Gamma_{t_{k-1}}(w_{t_k})\right] \\
&= \bigcap_{\psi \in \Psi(t_{k-1},w_{t_k})} \left(R_{t_{k-1},t_k}\left[-\sum_{I \in \mathcal{I}} \psi_I \Delta M_I(t_k) - \Gamma_{t_k}(w_{t_k})\right] + \xi(t_{k-1},w_{t_k},\psi) + \Gamma_{t_{k-1}}(w_{t_k})\right).
\end{align*}

	Let $\psi \in \Psi(t_{k-1},w_{t_k})$. By construction of the conditional halfspace $\Gamma_{t_{k-1}}(w_{t_k}) := \Gamma_{t_k}(w_{t_k}) \cap \lcal_{t_{k-1}}^d$ and the set-valued subtraction, $\Gamma_{t_{k-1}}(w_{t_k}) = \left(\Gamma_{t_k}(w_{t_k}) -^{\scalebox{0.5}{$\bullet$}} \Gamma_{t_k}(w_{t_k})\right) \cap \lcal_{t_{k-1}}^d$ since $\Gamma_{t_k}(w_{t_k})$ is a convex cone.  This allows us to recover the representation
	\begin{align*}
	& R_{t_{k-1},t_k}\left[-\sum_{I \in \mathcal{I}} \psi_I \Delta M_I(t_k) - \Gamma_{t_k}(w_{t_k})\right] + \xi(t_{k-1},w_{t_k},\psi) + \Gamma_{t_{k-1}}(w_{t_k})\\
	&=  G_E(t_{k-1},\psi,w_{t_k})\Delta t_k + \xi(t_{k-1},w_{t_k},\psi) + \sum_{I \in \mathcal{I}} \psi_I \Delta M_I(t_k) + \Gamma_{t_{k-1}}(w_{t_k}) - \sum_{I \in \mathcal{I}} \psi_I \Delta M_I(t_k)  \\
	&=  G_E(t_{k-1},\psi,w_{t_k})\Delta t_k +H_{t_k}(X,w_{t_k}) - \sum_{I \in \mathcal{I}} \psi_I \Delta M_I(t_k)  \\
	&=  G_E(t_{k-1},\psi,w_{t_k})\Delta t_k  + \left(H_{t_k}(X,w_{t_k}) -^{\scalebox{0.5}{$\bullet$}} \left[\sum_{I \in \mathcal{I}} \psi_I \Delta M_I(t_k) + \Gamma_{t_k}(w_{t_k})\right]\right) \cap \lcal_{t_{k-1}}^d\\
	&= G_E(t_{k-1},\psi,w_{t_k})\Delta t_k  + \left(\cl\left[\ycal(t_k) + \Gamma_{t_k}(w_{t_k})\right] -^{\scalebox{0.5}{$\bullet$}} \left[\sum_{I \in \mathcal{I}} \psi_I \Delta M_I(t_k) + \Gamma_{t_k}(w_{t_k})\right]\right) \cap \lcal_{t_{k-1}}^d\\
	&= G_E(t_{k-1},\psi,w_{t_k})\Delta t_k  + \left(\cl\left[\ycal(t_k) + \Gamma_{t_k}(w_{t_k})\right] -^{\scalebox{0.5}{$\bullet$}} \left[\sum_{I \in \mathcal{I}} \psi_I \Delta M_I(t_k) + \lcal^d_{t_k,+}\right]\right) \cap \lcal_{t_{k-1}}^d.
	\end{align*}
Combining the above calculations yields
\[
	\ycal(t_{k-1}) = \bigcap_{\substack{w_{t_k} \in W_{t_k}(X),\\\psi \in \Psi(t_{k-1},w_{t_k})}} \left[\begin{array}{l}G_E(t_{k-1},\psi,w_{t_k}) \Delta t_k \\ + \left(\cl\left[\ycal(t_k) + \Gamma_{t_k}(w_{t_k})\right] -^{\scalebox{0.5}{$\bullet$}}\left[\sum_{I \in \mathcal{I}} \psi_I \Delta M_I(t_k) + \lcal_{t_k,+}^d \right]\right) \cap \lcal_{t_{k-1}}^d \end{array}\right].
\]
Note that for $w_{t_k}\in \lcal^d_{t_k,+}\sm W_{t_k}(X)$, the definition in \eqref{const} gives $\Psi(t_{k-1},w_{t_k})=\emptyset$ as desired, since no constructor of $H_{t_k}(X,w_{t_k})$ exists in this case. Hence, with the convention that the intersection in the SV-BS$\Delta$E over an empty index set gives $\lcal_{t_k}^d$, it follows that the SV-BS$\Delta$E is satisfied.

The converse follows by the same logic as above since the construction of $\Psi(t_{k-1},w_{t_k})$ guarantees that it is part of a predictable representation of $H_{t_k}(X,w_{t_k})$.
\end{proof}

\begin{corollary}
	Fix some positive vector $\r \in \R^d_{++}$. Let $X \in \lcal_{t_K}^d$. 
	Let $(\ycal(t))_{t \in \mathbb{T}}$ be a process such that $\ycal(t_k) = R_{t_k}(X)$ for each $k \in \{0,\ldots,K\}$. Then, for each $k\in\cb{1,\ldots,K}$, there exists a nonempty set $W(t_k)\subseteq \lcal^d_{t_k}$ such that for each $w_{t_k}\in W(t_k)$, there exists a unique pair $(\hat{\xi}(t_{k-1},w_{t_k}),\hat{\psi}(t_{k-1},w_{t_k}))\in \lcal_{t_{k-1}}\times(\lcal_{t_{k-1}})^{\mathcal{I}}$ satisfying
	\begin{equation}\label{r-unique}
	 \left[\hat{\xi}(t_{k-1},w_{t_k})+ \sum_{I \in \mathcal{I}} \hat{\psi}_I(t_{k-1},w_{t_k}) \Delta M_I(t_k)\right]\r + \Gamma_{t_k}(w_{t_k}) = \cl[\ycal(t_k) + \Gamma_{t_k}(w_{t_k})],
	\end{equation}
	and there exists no such pair for each $w_{t_k}\in \lcal^d_{t_k,+}\sm W(t_k)$. Moreover, the reformulated SV-BS$\Delta$E
	\begin{align*}
	\ycal(t_{k-1}) &= \bigcap_{w_{t_k}\in W(t_k)} \left[\begin{array}{l} G_E(t_{k-1},\hat{\psi}(t_{k-1},w_{t_k})\r,w_{t_k}) \Delta t_k \\ + \left(\cl\left[\ycal(t_k) + \Gamma_{t_k}(w_{t_k})\right] -^{\scalebox{0.5}{$\bullet$}} \left[\sum_{I \in \mathcal{I}} \hat{\psi}_I(t_{k-1},w_{t_k}) \r \Delta M_I(t_k) + \lcal_{t_k,+}^d\right]\right) \cap \lcal_{t_{k-1}}^d \end{array}\right],\\
	\ycal(t_K) &= -X + R_{t_K}(0)
	\end{align*}
	holds.
\end{corollary}
\begin{proof}
	Let us take $W(t_k)=W_{t_k}(X)$, where $W_{t_k}(X)$ is defined by \eqref{Wdefn}.
	Fixing the direction $\r \in \R^d_{++}$ guarantees the uniqueness of the predictable representation of $H_{t_k}(X,w_{t_k})$ for $w_{t_k} \in W(t_k)$ in the form of \eqref{r-unique}. As noted in the proof of Proposition~\ref{lemma:SV-BSDE}, no predictable representation exists of $H_{t_k}(X,w_{t_k})$ for $w_{t_k}\in \lcal_{t_k,+}^d \sm W(t_k)$.  
It follows that the reformulated SV-BS$\Delta$E is equivalent to that given in Proposition~\ref{lemma:SV-BSDE}.
\end{proof}

\begin{remark}
	The SV-BS$\Delta$Es given in Proposition~\ref{lemma:SV-BSDE} and all subsequent results of this section can utilize the dual variables $(\Q,w_{t_{k-1}}) \in \W_{t_{k-1}}$ in place of $w_{t_k} \in \lcal_{t_k,+}^d$ by considering $w_{t_k} := w_{t_{k-1}}^{t_k}(\Q,w_{t_{k-1}})$.  This follows from an application of \citet[Lemma~A.1]{mptc}.
\end{remark}

\begin{remark}
	While in the BS$\Delta$I framework we introduce semi-local and local versions, this does not appear to be possible for the SV-BS$\Delta$E setup.  Specifically, the Minkowski difference restricted to an earlier time point as taken in the SV-BS$\Delta$E cannot readily be defined $\omega$-wise. Conceptually, this can be viewed as akin to the work of \cite{bentaharlepinette} insofar as the recursive formulation of multiportfolio time consistency is defined with respect to the selectors rather than the random sets.
\end{remark}

We conclude this section with a continuation of Examples \ref{exmodel1} and \ref{exmodel2} on the superhedging risk measure.

\begin{example}\label{ex:SV-BSDE}
	In the setting of Example \ref{exmodel1}, let us calculate the set-valued driver $G_E^{SHP}$ of the superhedging risk measure $R^{SHP}$ when $d=2$ and $m=1$.  As with Examples~\ref{exmodel2} and~\ref{exmodel3}, we can generalize to larger number of assets and random walks but will focus on this simple case for illustrative purposes.  Following \eqref{eq:SV-driver} and the definitions in Example \ref{exmodel1}, for each $t\in\cb{0,\ldots,T-1}$, $\psi\in\lcal_t^2$, and $w\in\lcal_{t+1,+}^2$, we have
	\begin{align*}
	&G_E^{SHP}(t,\psi,w)\\
	&=R^{SHP}_{t,t+1}[-\psi\Delta M(t+1) - \Gamma_{t+1}(w)]=SHP_t[\psi B(t+1) + \Gamma_{t+1}(w)]\\
    &=\cb{Y\in\lcal_t^2 \; | \; Y\in \psi B(t+1) + u + \sum_{s=t}^T L_d^0(\F^M_s,K_s), \, u\in\lcal_{t+1}^2, \, w^\T u \geq 0}\\
    &=L_d^0(\F^M_t,K_t)+\cb{Y\in\lcal_t^2 \; | \; Y\in \psi B(t+1) + u + \sum_{s=t+1}^T L_d^0(\F^M_s,K_s), \, u\in\lcal_{t+1}^2, \, w^\T u \geq 0}\\
    &=L_d^0(\F^M_t,K_t)+\Bigg\{Y\in\lcal_t^2 \; | \; \begin{pmatrix} Y^1\\Y^2\end{pmatrix}=\begin{pmatrix}\psi^1\\ \psi^2\end{pmatrix}B(t+1)+\begin{pmatrix}u^1 \\ u^2\end{pmatrix} + \sum_{s=t+1}^T \sqb{b_s\begin{pmatrix} -S_s^b\\1\end{pmatrix}+ a_s\begin{pmatrix} S_s^a\\-1\end{pmatrix}},\\
	&\quad\quad\quad\quad\quad\quad\quad\quad\quad \quad\quad\quad\quad    a_s,b_s\in \lcal_{s,+}^2, s\in\cb{t+1,\ldots,T}, \, u\in\lcal_{t+1}^2, \, w^1 u^1 + w^2 u^2 \geq 0\Bigg\}.
	\end{align*}
	Recall that, as provided in Example \ref{exmodel1} and presented in Example \ref{exmodel2}, for each $s\in\cb{t+1,\ldots,T}$, we have
	\[
	S_s^b=S_t^b e^{\sigma B(t+1)}e^{\sigma(M(s)-M(t+1))},\quad S_s^a=S_t^a e^{\sigma B(t+1)}e^{\sigma(M(s)-M(t+1))}.
	\]
    As stated previously, unlike in the BS$\Delta$I setup of Section~\ref{sec:bsdi}, no \emph{local} driver can be given for the SV-BS$\Delta$E formulation of the superhedging risk measure.
\end{example}


\section{Discussion}

In this work, we introduced two backward representations for multiportfolio time consistent dynamic set-valued risk measures: a BS$\Delta$I and a SV-BS$\Delta$E.  Though both of these representations provide an equivalent dynamic risk measure, these formulations provide important insights for considering dynamic risk measures in continuous time $\mathbb{T} = [0,T]$.  Specifically, in continuous time, either a backward stochastic differential inclusion (BSDI) or set-valued backward stochastic differential equation (SV-BSDE) could, potentially, be used to characterize a dynamic risk measure.  This has yet to be examined in the literature.
Our work on studying difference inclusions and equations provides the initial insights for when these concepts are appropriate to be applied. As such, we gain knowledge of the (likely) best approach for studying risk measures in continuous time.  In fact, based on the prior analysis, it is the opinion of the authors that that BSDIs are the proper methodology to consider. While presented only for risk measures, the results presented provide insights for the dynamic programming principle in multivariate problems more generally.

By studying both the BS$\Delta$I and SV-BS$\Delta$E, we begin to understand when these approaches are appropriate.  Namely, the inclusion appears to be the appropriate method if we care specifically about singular paths.  For instance, risk measures are used to compute capital requirements.  Ultimately, the singular, implemented capital investment over time is the important result for a practitioner rather than the entire set of acceptable requirements.  In contrast, a set-valued equation appears to be the appropriate method if we care about the ``mass'' of the set itself over time, rather than any specific value in that set. Such concepts are important beyond the immediate study of risk measures, for instance, for the mean-risk problem in \cite{KR20}, where the dynamic programming principle holds for the multi-objective version but \emph{not} the traditional scalar approach.  These results indicate that the BSDI is likely to be the appropriate approach for that problem; this is left for future study.

Following from the interpretation of risk measures, we talk about the set of acceptable capital requirements.  Thus, the individual requirements are the important notion themselves in our setting.  This leads to the notion that BSDIs appear to be the appropriate methodology for us to consider.  Though our recursive formulation of multiportfolio time consistency and the dynamic programming principle are defined with respect to the risk measure of sets (i.e., $R_t(X) = R_{t,s}[-R_s(X)]$), this is ultimately defined element-wise.  We conjecture that SV-BSDEs would be the only available method if set-valued portfolios (e.g., \cite{setportfolio}) or, generally, functions of sets were considered themselves.

While we are able to construct a SV-BS$\Delta$E for this setting, it does not permit an immediate integral (summation in the discrete setting) representation when incrementing from time $t$ to $T$ directly.  
To pass to the continuous time BSDI limit of the BS$\Delta$I, we propose following the approach of \cite{stadje} in which the risk measures are ``scaled and tilted'' first.  We conjecture that convergence of the BS$\Delta$I should be sought for the space of reachable sets as constructed in Corollary~\ref{multistepcor}; this is in contrast to studying the individual paths.
However, we leave that consideration for future works.

\section*{Acknowledgments}

The authors would like to thank two anonymous referees and the associate editor for their detailed feedback which was helpful in improving the manuscript. This work was initiated and an early version of it was presented at the two meetings on Dynamic Multivariate Programming (2018, 2019) hosted and funded by the Vienna University of Economics and Business. The authors would like to thank Birgit Rudloff for organizing these meetings as well as all the participants for fruitful discussions.

\bibliographystyle{named}

\end{document}